\documentclass[aps,prl,reprint,superscriptaddress,twocolumn,showkeys,amsmath,amssymb]{revtex4-1}
\usepackage[english]{babel}
\usepackage{microtype}
\usepackage[utf8]{inputenc}
\usepackage{amsthm}
\usepackage{braket}
\usepackage{mathtools}
\usepackage{dsfont}
\usepackage{mathrsfs}
\usepackage{braket}
\usepackage{hyperref}
\usepackage[shortlabels]{enumitem}
\usepackage{xcolor}

\newtheorem{lemma}{Lemma}
\newtheorem{theorem}{Theorem}

\def\tcm{T.C.M. Group, Cavendish Laboratory, University of Cambridge, J.J. Thomson Avenue, Cambridge, CB3 0HE, UK}
\def\DAMTP{DAMTP, University of Cambridge, Wilberforce Road, Cambridge, CB3 0WA, UK}

\begin{document}
\title{Fermion-Parity-Based Computation and its Majorana-Zero-Mode Implementation}
\author{Campbell K. McLauchlan}
\affiliation{\DAMTP}
\author{Benjamin Béri}
\affiliation{\DAMTP}
\affiliation{\tcm}
\date{October 2021}
\begin{abstract}
    Majorana zero modes (MZMs) promise a platform for topologically protected fermionic quantum computation.
    However, creating multiple MZMs and generating (directly or via measurements) the requisite transformations (e.g., braids) pose significant challenges.
    We introduce fermion-parity-based computation (FPBC): a measurement-based scheme, modeled on Pauli-based computation, that uses efficient classical processing to virtually increase the number of available MZMs 
    and which, given magic state inputs, operates without transformations.
    FPBC requires all MZM parities to be measurable, but this conflicts with constraints in proposed MZM hardware. 
    We thus introduce a design in which all parities are directly measurable and which is hence well suited for FPBC.
    While developing FPBC, we identify the ``logical braid group" as the fermionic analog of the Clifford group. 
\end{abstract}

\maketitle
Pauli-based computation (PBC) is an intriguing measurement-based alternative to the circuit model of quantum computing~\cite{Bravyi_PBC2016}.
By performing only a minimal number of adaptive Pauli measurements on ``magic state" inputs, PBC allows one to virtually expand the number of qubits in a quantum computer and forego the need to perform Clifford gates~\cite{Gottesman_thesis}, at the cost of efficient classical processing. 
While PBC is formulated for qubits, quantum computing can also use fermionic modes~\cite{Bravyi_Kitaev_Fermionic_QC2002}. 
Fermionic quantum computing is better suited to certain tasks, a notable example being many-electron, including quantum chemistry, simulations~\cite{Bravyi_Kitaev_Fermionic_QC2002,Maj_Based_FQC2018}.
For the fermionic hardware, Majorana zero modes (MZMs) are a promising option, as they offer \mbox{topological protection of quantum information~\cite{Kitaev_chain2001,Kitaev_Anyons2006, MZM_TQC_Review2015,Milestones_Paper2016,Karzig_Majorana2017,Maj_Based_FQC2018}.}

In this work, we formulate fermion-parity-based computation (FPBC), a fermionic counterpart of PBC, and propose a MZM hardware design well-suited to its implementation. 
En route, we identify the ``logical braid group" as the group of all Clifford-like fermionic gates.
For MZM computing, FPBC does not just mean fewer MZMs: it is a new computational model, distinct from the circuit model of previous measurement-based approaches~\cite{Mmt_Only_QC_2008, Mmt_only_no_forced_mmts_2017,Karzig_Majorana2017,Mmt_only_QC_corner_modes_2020,Optimizing_Clifford_Gates_2020}, that eliminates the need to generate braiding and other Clifford-like transformations~\cite{Mmt_Only_QC_2008,Majoranas_Tjunctions2011,Coulomb_assisted_braiding2012,Milestones_Paper2016,Mmt_only_no_forced_mmts_2017,Karzig_Majorana2017,Mmt_only_QC_corner_modes_2020,Optimizing_Clifford_Gates_2020},
and thus avoids the \mbox{associated overheads~\cite{Time_Scales_Coulomb_Block2016,Diabatic_errors2016,Optimizing_Clifford_Gates_2020}.} 

A key requirement for FPBC is to be able to measure potentially complicated strings of MZMs. 
We find that configuration constraints present obstacles to this in existing MZM designs.
Our design, based on top-transmon ingredients~\cite{Transmon_Qubit2007, Transmon_Experiment2008, Top_Transmon2011, Majorana_Ram_Flux_control2013}, is free of such constraints.
Furthermore, unlike circuit-based computing in existing designs~\cite{Majorana_Ram_Flux_control2013,Karzig_Majorana2017, Maj_Based_FQC2018, Mmt_Only_QC_2008, Mmt_only_no_forced_mmts_2017, Mmt_only_QC_corner_modes_2020, Optimizing_Clifford_Gates_2020}, FPBC with our design uses no ancilla MZMs.
The only remaining limitation is locality, as we shall explain.

\textit{Fermionic Quantum Computing and Logical Braids:}
Consider $2n$ Majorana operators $\gamma_j=\gamma_j^\dagger$
($j=1,\ldots,2n$) with anti-commutator $\{\gamma_j,\gamma_k\}=2\delta_{jk}$. 
These $2n$ modes have total fermion parity $\Gamma_{2n} = i^n\prod_{j=1}^{2n}\gamma_j$. 
Let Maj$(m)$ denote the group of Majorana strings generated by $\gamma_1,\ldots,\gamma_m$ and the phase factor $i$,  
and $\overline{\text{Maj}}(m)$ denote the subgroup of Maj$(m)$ that commutes with $\Gamma_{2n}$.
We call the Hermitian elements of $\overline{\text{Maj}}(m)$ fermion parity operators; FPBC will be based on their adaptive measurements. 

To develop FPBC, we first consider fermionic quantum computing in the circuit model, and then, analogously to PBC~\cite{Bravyi_PBC2016,Mithuna_Magic_State_PBC2019},
show how FPBC can simulate it. 
We consider fermionic circuits based on the universal gate set
$\{ W_{4,abcd}=\exp{(i\frac{\pi}{4}\gamma_a\gamma_b\gamma_c\gamma_d)},\,T_{2,ab}=\exp{(\frac{\pi}{8}\gamma_a\gamma_b)}\}$~\cite{Bravyi_Kitaev_Fermionic_QC2002}.
Note that non-commuting $W_4$ operators can generate all possible gates of the form $W_{2k,i_1i_2\ldots i_{2k}}=\exp{(\pm i^{k+1}\frac{\pi}{4}\gamma_{i_1}\gamma_{i_2}\ldots \gamma_{i_{2k}})}$ (including braid operators $W_{2,ab}$~\cite{ReadGreen2000,MZMs_in_pwave_superconds2001,Non-Abelian_TQC_Review,MZM_TQC_Review2015}). 
This is because $W_{4,abcd}$ is a ``logical braid" between $\gamma_a$ and $i\gamma_b\gamma_c\gamma_d$, the latter being a ``logical Majorana" relative to $\gamma_a$
(i.e., a parity-odd, Hermitian Majorana string anti-commuting with $\gamma_a$~\cite{LogicalMajorana}),
and hence can send $W_{2k} \mapsto W_{2k\pm2}$ under conjugation.
Due to this observation we refer to the group generated by the $W_4$ as the logical braid group and its elements, including all $W_{2k}$, as logical braids.
Under conjugation, $W_4$ gates map between strings in $\text{Maj}(2n)$~\cite{Bravyi5p2_2006}; they are Clifford-like. Indeed, logical braids are the only parity-preserving unitaries with this property~\cite{Appfn}.

A key insight for FPBC is that a $T_2$ gate can be implemented via a ``magic state gadget."
Here, we describe this procedure using a dense encoding~\cite{MZM_TQC_Review2015} of magic states, which is more suitable for our fermionic hardware (cf. below) and a more efficient use of quantum resources~\cite{Maj_Based_FQC2018}.
To implement $t$ $T_2$ gates, assume we have a separate register $\mathcal{R}_t$ of $2t+2$ Majoranas with its own conserved parity $\Gamma_{2t+2}$.
Define two sets of %$t$ mutually commuting 
operators in $\overline{\text{Maj}}(2t+2)$: $\lbrace X_{1},X_{2},\ldots,X_{t}\rbrace$ and $\lbrace s_j=i\gamma_{2j-1}\gamma_{2j}|j=1,\ldots, t\rbrace$,
obeying $\{s_j,X_{j}\}=0$ (for all $j$) and $[s_{j^\prime},X_{j}]=[X_{j^\prime},X_j]=[s_{j^\prime},s_j]=0$ ($j'\neq j$).
Then let register $\mathcal{R}_t$ be in the state $|\psi^{(t)}\rangle=T_{2,12}T_{2,34}\ldots T_{2,2t-1\,2t}|\psi^X\rangle$,
where $|\psi^X\rangle$ is the $+1$-eigenstate of all $X_j$ operators. Thus, the register contains $t$ magic states densely encoded into $2t+2$ Majoranas.
The gate $T_{2,ab}=\exp{(\frac{\pi}{8}\gamma_a\gamma_b)}$ can then be applied to Majoranas $a,b$ in a separate register $\mathcal{R}_n$ with its own conserved parity, using the procedure or ``gadget" (shown in Fig.~\ref{fig:circuit}):
$\mathcal{M}_{j,ab}=R_j\Pi^{m_j}_{is_j\gamma_a\gamma_b}$, for $j\in\lbrace 1,\ldots,t\rbrace$,
enacted on both registers.
Here, $\Pi^{m_j}_{is_j\gamma_a\gamma_b}$ is the projector representing the measurement of $is_j\gamma_a\gamma_b$ with outcome $m_j$, and $R_j=[\exp(\frac{\pi}{4}\gamma_b\gamma_a)]^{(1+m_j)/2}\exp(\frac{\pi}{4}\gamma_a\gamma_bX_{j})$ is a measurement-dependent logical braid.

Magic states can be distilled from multiple approximate copies with logical braids and measurements, using magic state distillation~\cite{Magic_State_Dist,Bravyi5p2_2006,MSD_Experiment2011,Maj_Based_FQC2018} -
this is one of the leading candidates for preparing high-fidelity magic states in Majorana-based architectures, and thereby for achieving fault-tolerant, universal quantum computation~\cite{Maj_Based_FQC2018,Karzig_Majorana2017,QC_with_MFCs,MZM_TQC_Review2015,Maj_Triangle_Code2018,Dark_Space_Stab2020}.
Much work has been devoted to optimising its resource cost~\cite{MSD_Low_Overhead2012,MSD_Low_Overhead2013,MSD_Low_Overhead2017,MSD_Not_Costly2019} and finding alternatives that can also be used to prepare magic states~\cite{Karzig_geom_magic_2016,Karzig_geom_magic_meas2019}. 

\begin{figure}[t]
    \centering
    \includegraphics[width=\linewidth]{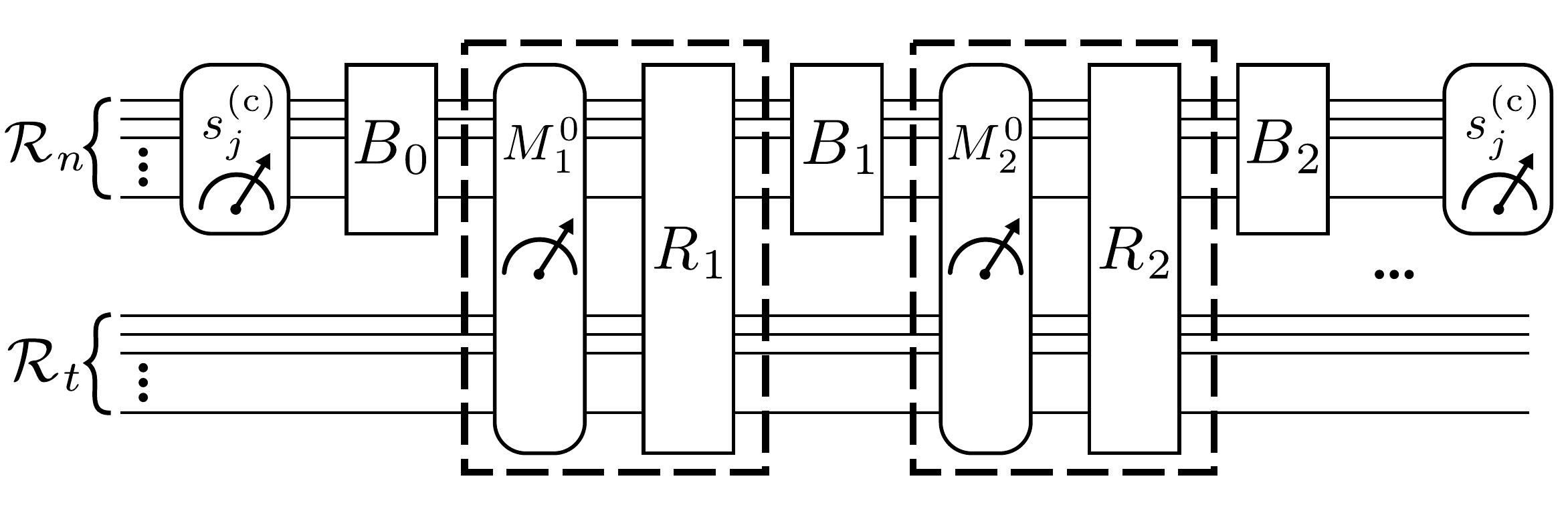}
    \caption{An arbitrary fermionic circuit $C$ on register $\mathcal{R}_n$, to be simulated by FPBC. 
    The $T_2$ gates are enacted via magic state gadgets (dashed boxes), with magic states encoded in register $\mathcal{R}_t$. 
    The gadgets involve a two-register fermion parity measurement $M_j^0$
    and a measurement-dependent logical braid $R_j$.
    $C$ involves $t$ uses of the gadget, interspersed with logical braids $B_i$, and ends with the measurement of all $s_j^\text{(c)}$. 
    Dummy measurements of all $s_j^\text{(c)}$ are appended to the start of $C$. 
    }
    \label{fig:circuit}
\end{figure}

\textit{Fermion-Parity-Based Computation:} 
By performing adaptive fermion parity measurements only on $\mathcal{R}_t$ initialized in state $|\psi^{(t)}\rangle$, and efficient classical processing, FPBC can simulate an arbitrary fermionic circuit $C$ on $\mathcal{R}_n$.
Without loss of generality, we take $C$ to act on $\mathcal{R}_n$ initialized in the $+1$-eigenstate of all $s_j^\text{(c)}$ [$j=1,\ldots,n$; (c) indicates $\mathcal{R}_n$ operators], and that it uses $t=\text{poly}(n)$ $T_2$ gates, interspersed with logical braids $B_i$ (recall that $T_2$ gates and logical braids form a universal gate set~\cite{Bravyi_Kitaev_Fermionic_QC2002}).
$C$ ends by measuring all $s_j^\text{(c)}$ on $\mathcal{R}_n$, i.e., sampling the output distribution. The bit string $\mathbf{b}$ of these final measurement results comprises the output of the circuit.

The first step towards simulating $C$ by FPBC is to replace all $T_2$ gates with magic state gadgets.
As shown in Fig.~\ref{fig:circuit}, $C$ then involves logical braids $B_i$ and $R_i$, fermion parity measurements (labeled $M_i^0$ for $i=1,\ldots, t$) from the $t$ uses of the gadget, 
and final $s_j^\text{(c)}$ measurements.
We denote these final measurements $M_{t+j}^0 \equiv s_j^\text{(c)}$ ($j=1,\ldots, n$).
The next step is to eliminate all logical braids, by commuting the $B_i$ and $R_i$ to the end of the circuit, thereby updating $M_i^0\mapsto M_i\in \overline{\text{Maj}}(2(n+t+2))$. Since the quantum state after the final measurement is discarded, the logical braids now have no effect on the output, and can be deleted. 
For what follows, we append a set of dummy measurements of all $s_j^\text{(c)}$ to the start of $C$, shown in Fig.~\ref{fig:circuit}, which have outcomes $+1$ on $\mathcal{R}_n$'s initial state, and define $M_{j-n}\equiv s_j^\text{(c)}$ for $j=1,\ldots,n$.
At this stage, either $[M_i,M_j]=0$ or $\{M_i,M_j\}=0$ for all $i,\, j$.
We now show that one can limit the measurements to a mutually commuting set, thereby reducing
the number needing to be performed and restricting the computation to $\mathcal{R}_t$.
To achieve this, we go through the $M_i$ sequence, starting with $M_1$, and, if we reach an $i$ such that $\{M_i,M_j\}=0$ for some $j<i$, we delete $M_i$ and replace it with the logical braid
\begin{equation}
V(\lambda_i,\lambda_j) = \frac{\openone + \lambda_i\lambda_j M_i M_j}{\sqrt{2}} = \exp\left(\frac{\pi}{4} \lambda_i\lambda_j M_iM_j\right),
\end{equation}
where $\lambda_j = \pm 1$ is the measurement outcome of $M_j$ and $\lambda_i=\pm1$ is chosen uniformly at random. 
As in PBC~\cite{Bravyi_PBC2016,Mithuna_Magic_State_PBC2019}, this simulates the measurement of $M_i$: 
$\{M_i,M_j\}=0$ implies equal measurement probabilities $1/2$ for $M_i$, which is simulated by uniformly choosing $\lambda_i$ at random, and since $\lambda_j M_j=\openone$ on the pre-measurement state, $V(\lambda_i,\lambda_j)$ produces the correct corresponding post-measurement state.
We then commute $V(\lambda_i,\lambda_j)$ past all $M_{l>i}$. Again, it can then be deleted. 
(Henceforth we leave the resulting updates of $M_{l>i}$ implicit.)
For the final $n$ measurements, if $M_i$ is replaced by its corresponding $V(\lambda_i,\lambda_j)$, we include the classically randomly generated value of $\lambda_i$ in $\mathbf{b}$.

Finally, we are left with a sequence of mutually commuting $M_i$. 
For $j\leq 0$ we still have dummy measurements $M_j=s_j^\text{(c)}$ and $\lambda_{j}=1$, which completely specifies a basis of $\mathcal{R}_n$ (within a given parity sector). 
Hence, since $[M_1,M_j]=0$ for all $j\leq 0$, we can restrict $M_1$ to $\mathcal{R}_t$ without changing its measurement distribution or post-measurement state. 
We can then restrict $M_2$ to $\mathcal{R}_t$, since $[M_2,M_j]=0$ for all $j\leq 1$, and so on.
Doing this for all $M_{j>0}$ and then discarding $M_{j\leq 0}$, we thereby restrict the entire computation to $\mathcal{R}_t$.
There are only $t$ independent commuting parities (besides $\Gamma_{2t+2}$) on $\mathcal{R}_t$. 
Using efficient classical computation~\cite{nielsen_chuang}, one computes the outcomes for those $M_j$ dependent on preceding $M_i$, and deletes them. 
The quantum part of the computation is thus reduced to the adaptive measurement of $p\leq t$ mutually commuting parities on $\mathcal{R}_t$.
The remaining entries in $\mathbf{b}$ (those not filled by the classically sampled $\lambda_i$) come from the outcomes of those $M_{j>t}$ that were not replaced by logical braids; via the process described above these outcomes are either measured explicitly or computed classically.
Thus, assisted with poly$(n)$-time classical processing, we can sample from $C$'s output distribution using FPBC.

\textit{FPBC Hardware:} 
To perform FPBC, one needs hardware such that the fermion parities $M_i$ on $\mathcal{R}_t$ are measurable. 
In existing MZM designs, one can measure only those $M_i$ that meet certain configuration constraints.
For  example, in Majorana transmon setups~\cite{Majorana_Ram_Flux_control2013,Milestones_Paper2016,Maj_Based_FQC2018,Smith_Bartlett_Readout2020}
one has ``readout islands" with a pair of MZMs on each, and only those $M_i$ are measurable that feature no MZM without its readout-island pair.
Magic state gadgets in these setups, however, require inter-island logical braids and/or measurements,
which can generate FPBCs with unmeasurable $M_i$~\cite{Appfn}.
(Subsequent braids may bring $M_i$ to a measurable configuration; however in typical setups, and for large $t$, only for a vanishingly small proportion of $M_i$ does just a constant-in-$t$ number of such braids suffice~\cite{Appfn}.)

\begin{figure}
    \centering
    \includegraphics[width=\linewidth]{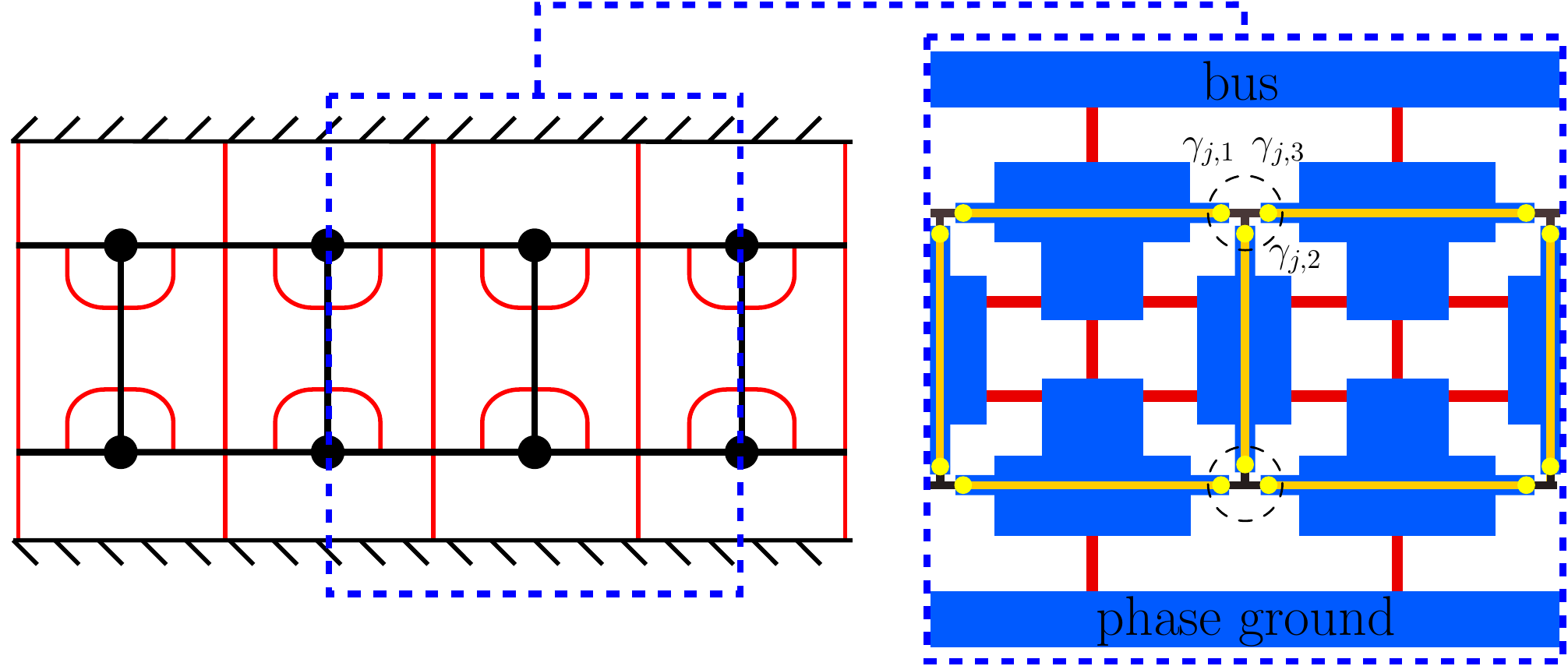}
    \caption{Section of FPBC hardware.
    Left: top and bottom black dashed regions are superconducting plates - the bus and phase ground respectively. 
    Thick black lines correspond to nanowire-hosting superconducting islands while black dots indicate tri-junctions between nanowires. 
    The design may be continued to the right and left.
    Right: more detailed illustration of the region indicated. 
    Superconducting plates and islands are shown in blue.
    Nanowires (yellow) host Majorana bound states (labelled $1$, $2$ and $3$)
    at their ends, combining to form a single Majorana zero mode at each tri-junction (dashed circles).
    In both panels, tunable Josephson junctions are indicated with red lines.}
    \label{fig:Model}
\end{figure}

We introduce a design (sketched in Fig.~\ref{fig:Model}) that is free of such configuration constraints.
The core ingredients and the corresponding physical considerations are based on Refs.~\onlinecite{Top_Transmon2011,Majorana_Ram_Flux_control2013}.
The MZMs appear at tri-junctions between Majorana bound states at the ends of spin-orbit nanowires on superconducting islands~\cite{Kitaev_chain2001, von_Oppen_Majorana_Nanowires2010, Das_Sarma_Majorana_Nanowire_2010, Alicea_Review_2012, Nature_Maj_Experiment2012, Science_Maj_Experiment_2012, Fractional_AC_Josephson_Experiment_2012, Maj_Experiment_Nano_Lett_2012,  Beenakker_Review_2013, Maj_Epitaxy_2015, Albrecht-et-al-exponential-2016, Maj_Experiment_2_2017, Lutchyn_Review_2018, vaitiekenas2021zero,Nanowire_Review_2021}.
The islands are connected via tunable Josephson junctions (JJs) to other islands and, for some islands, also to one of two superconducting plates, called the bus and phase ground.
This entire system is enclosed within a transmission line resonator. 
As we next explain, this has a parity-dependent resonance frequency, which allows one to measure the $M_i$ via dispersive readout~\cite{Top_Transmon2011, Transmon_Qubit2007,Coulomb_assisted_braiding2012,Maj_Circuit_QED2015,Smith_Bartlett_Readout2020,Dispersive_Regime_2007}.
[The similar parity dependence of the transmon groundstate can be used to implement (approximate) $T_2$ gates~\cite{Top_Transmon2011,Majorana_Ram_Flux_control2013} and hence to supply (noisy) magic states for distillation.]

A JJ between superconductors $a$ and $b$, with phases $\phi_a$ and $\phi_b$ of their superconducting order parameters respectively, contributes a term $E_{J,ab}(1- \cos{(\phi_a-\phi_b)})$ to the Hamiltonian~\cite{Tinkham_Superconductivity}, for some energy $E_{J,ab}$ that can be controlled by 
fluxes or electrostatic gates~\cite{Majorana_Ram_Flux_control2013,Gate_Controlled_JJ_1_2015, Gate_Controlled_JJ_2_2015}. 
By tuning these control parameters, each JJ can thus be turned on or off, corresponding to Josephson energy $E_{J,ab}^{\text{(on/off)}}$, %between superconductors
where $E_{J,ab}^{\text{(on)}} \gg E_{J,ab}^{\text{(off)}}$.
The $k$\textsuperscript{th} island has charging energy scale $E_{C,k} = e^2/2C_k$ for total capacitance $C_k$ between island $k$ and all other superconductors to which it is connected.
We take $E_{J,ab}^{\text{(on/off)}}$ to be of the same order of magnitude for all $ab$ and similarly for $E_{C,k}$ across all $k$. 
In what follows, each island will be connected (directly or via a path of ``on" JJs) to either the bus or phase ground; we call these bus-connected and ground-connected islands, respectively. 
We assume that the Josephson energy dominates for all islands, namely that $E_{J,ak}^{\text{(off)}}\gtrsim E_{C,k}$ for all islands $a,k$ with JJs connecting them.
Given this, and that $E_{J,ak}^{\text{(on)}}/E_{C,k} \gg E_{J,bl}^{\text{(off)}}/E_{C,l}$ 
(for all $a,k; b,l$ with JJs),
any bus-connected (ground-connected) island has superconducting phase pinned to that of the bus (phase ground)~\cite{Majorana_Ram_Flux_control2013}.
Hence, we can view the entire system as having a single effective JJ between bus- and ground-connected subsystems.
The corresponding Josephson and charging energies are $E_J$ and $E_C$, respectively, associated to sums of (``off"-state) Josephson energies and capacitances between the bus- and ground-connected subsystems.
We will take $E_J \gg E_{C}$, i.e., work in the transmon regime~\cite{Transmon_Qubit2007}.

The $j^\text{th}$ tri-junction has  Hamiltonian~\cite{Coulomb_assisted_braiding2012,Majorana_Ram_Flux_control2013,Majoranas_Tjunctions2011,SauClarkeTewari2011}
\begin{equation}
\!\!V_{M,j}=\frac{E_{M}}{2}\!\!\!\sum_{a,b,c=1}^{3}\!\!\epsilon_{abc}A_{j,a}(i\gamma_{j,b}\gamma_{j,c})=iE_M |\mathbf{A}_j| \gamma_{j,+}\gamma_{j,-}.
\end{equation}
Here $\gamma_{j,1}$, $\gamma_{j,2}$, and $\gamma_{j,3}$ are the Majorana bound states at the ends of the nanowires at the $j^\text{th}$ tri-junction and $E_M$ is the overall tri-junction energy scale. 
The $A_{j,a}$ include phase-dependent cosines encoding the $4\pi$-periodic Josephson effect~\cite{Kitaev_chain2001,Stab_4pi_Josephson_Effect_2011} (cf. the flux-dependent couplings of Refs.~\onlinecite{Coulomb_assisted_braiding2012,Majorana_Ram_Flux_control2013}) and $|\mathbf{A}_j|^2 = \sum_{a=1}^3A_{j,a}^2$. 
The coupling of the three Majorana bound states results in a MZM which we denote $\gamma_{j,0}$, and two more Majorana modes $\gamma_{j,+}$ and $\gamma_{j,-}$ encoding a nonzero-energy fermion~\cite{Appfn}. 

We take $E_M\ll\hbar \Omega_0$, where $\Omega_0\approx \sqrt{8E_JE_C}/\hbar$ sets the transmon level spacing~\cite{Transmon_Qubit2007}; the system is thus a top-transmon perturbed by the $V_{M,j}$~\cite{Appfn}. 
For low-lying levels, $V_{M,j}$ can be taken at zero bus-ground phase difference~\cite{Majorana_Ram_Flux_control2013}.
Without $V_{M,j}$, the effect of Majorana bound states is a contribution $(-1)^m\delta\varepsilon_m \mathcal{P}$ to the $m$\textsuperscript{th} transmon level energy, where  $\delta\varepsilon_m\propto \exp(-\sqrt{8E_J/E_C})$, and $\mathcal{P}$ is the joint fermion parity of Majorana bound states on bus-connected islands~\cite{Majorana_Ram_Flux_control2013,Appfn}. 
In considering $V_{M,j}$, we work with $E_M \gg \delta\varepsilon_m$ and to first order in $\delta\varepsilon_m/E_M$. 
This allows one to project $\mathcal{P}$ to $Q=P_-\mathcal{P}P_-^\dagger$, where $P_- = \prod_j P_{j,-}$ with $P_{j,-}=(1-i\gamma_{j,+}\gamma_{j,-})/2$.

Dispersive readout thus measures $Q$. 
The only fermion operators contributing to $Q$ are the $\gamma_{j,0}$, with $\gamma_{j,0}$ entering $Q$ if and only if there are an odd number of bus-connected islands around tri-junction $j$. 
With one bus-connected island at tri-junction $j$, only the $\gamma_{j,a}$ on that island features in $\mathcal{P}$; 
then the projection gives $ P_- \gamma_{j,a}P_-^\dagger = A_{j,a}\gamma_{j,0}/|\mathbf{A}_j|$~\cite{Appfn}.
For three bus-connected islands at tri-junction $j$, all three $\gamma_{j,a}$ feature in $\mathcal{P}$; we have 
$P_-(i\gamma_{j,a}\gamma_{j,b}\gamma_{j,c})P_-^\dagger = P_-(i\gamma_{j,0}\gamma_{j,+}\gamma_{j,-})P_-^\dagger = -\gamma_{j,0}$. 
With two bus-connected islands, tri-junction $j$ contributes a scalar factor to $Q$: 
$P_-(i\gamma_{j,a}\gamma_{j,b})P_-^\dagger = - \sum_c\epsilon_{abc}A_{j,c}/|\mathbf{A}_j|$. 
For a given configuration of bus- and ground-connected islands, and focusing on the lowest two transmon levels ($m=0,1$), $Q$ can be measured via the shift
\begin{align}\label{eqn:Shift_freq}
    \omega_{\text{shift}} = \frac{C}{2} (\delta\varepsilon_1 + \delta\varepsilon_0) \prod_{\substack{j\,|\,\text{1, 2 islands}\\ \text{bus-connected}}}\frac{A_{j,\alpha_j}}{|\mathbf{A}_j|}
\end{align}
in the resonator's resonance frequency upon flipping $Q$'s eigenvalue.
Here, $C$ is a constant dependent on transmon and resonator parameters~ \cite{Appfn}, the product runs over tri-junctions around which one or two islands are bus-connected, and $\alpha_j$ is set by the $j^\text{th}$ tri-junction's bus-connected island configuration.

\begin{figure}[t]
    \centering
    \includegraphics[width=\linewidth]{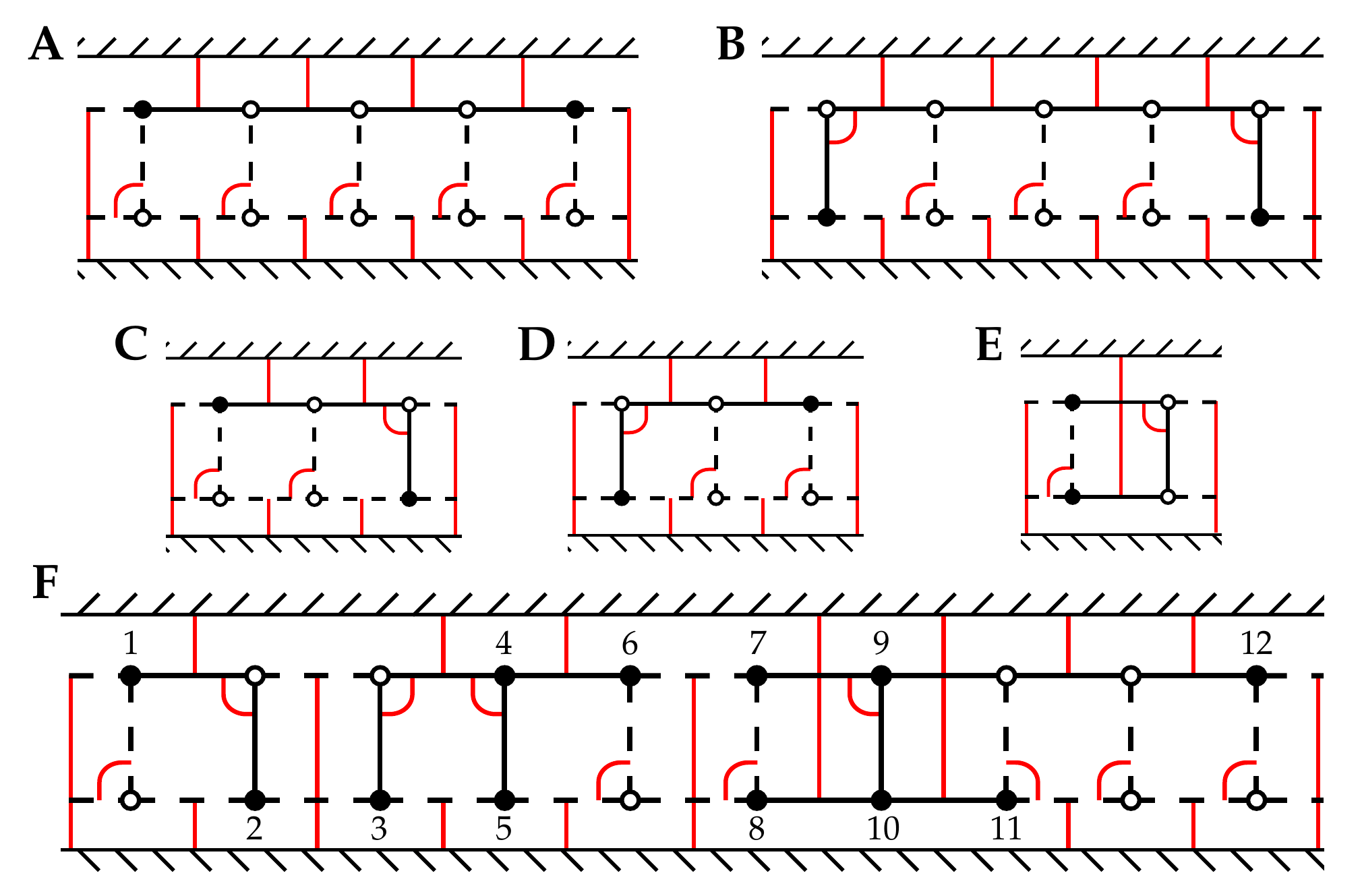}
    \caption{The configuration for measuring any parity operator $M$ is obtained by combining five basic paths (labelled A-E) of bus-connected islands; 
    these are the shortest clockwise paths connecting pairs of MZMs. 
    Paths A-D have variable length while Path E has a fixed length.
    Solid (dashed) lines in all panels indicate bus-connected (ground-connected) islands. Filled (unfilled) dots are MZMs that do (do not) feature in $M$. 
    Only ``on" Josephson junctions are indicated (red lines); others are omitted for clarity.
    As an example, Panel F shows the measurement configuration for a $12$-MZM parity operator. 
    The MZMs are indexed as in the main text and basic paths connect MZMs $j$ and $j+1$ for odd $j$. 
    When combining the basic paths, precisely those islands belonging to an odd number of paths are bus-connected.    
    }
    \label{fig:Algorithm_2_nodes}
\end{figure}

\textit{Arbitrary parity measurement:} 
The preceding discussion hints that our design allows for the measurement of any MZM parity $M_i$. We now explain this in detail.
Since $\gamma_{j,0}$ features in $Q$ when an odd number of islands surrounding it are bus-connected, precisely those $\gamma_{j,0}$ that are endpoints of a path of bus-connected islands feature in $Q$ (see Fig.~\ref{fig:Algorithm_2_nodes}).
We convert this observation into the following prescription: 
Let $M_i$ feature those $\gamma_{j,0}$ with $j$ in some set $S_{M_i}$. 
Index the labels $j\in S_{M_i}$ with $k_j=1,\ldots , |S_{M_i}|$ such that $k_{j'}<k_j$ if $\gamma_{j,0}$ is to the right of or directly below  $\gamma_{j',0}$ (cf. Fig.~\ref{fig:Algorithm_2_nodes}F). 
Pair the $\gamma_{j,0}$ with successive $k_j$ (i.e., first with second, third with fourth, etc.) and, for each pair, draw the shortest clockwise path of islands between the two MZMs. 
We then connect all islands featuring in an odd (even) number of paths to the bus (phase ground).

The measurement configuration thus formed for $M_i$ is realizable with the JJs indicated in Fig.~\ref{fig:Model}.
The shortest clockwise path between a MZM pair is one of five basic paths shown in Fig.~\ref{fig:Algorithm_2_nodes}A-\ref{fig:Algorithm_2_nodes}E.
A combination of these is realizable if there exists a path through ``on" JJs from every bus-connected (ground-connected) island to the bus (phase ground), and only ``off" JJs link bus-connected and ground-connected subsystems. 
In Fig.~\ref{fig:Algorithm_2_nodes} we indicate how the JJs achieve this for each basic path.
All pairs of basic paths are trivially realizable if we omit Path E, since then no bottom-row horizontal island is bus-connected, and bus-connected vertical islands are always adjacent to a bus-connected horizontal island.
There are a further five pairs that include Path E (EI for I$=$A$,\ldots,$E) which all can be checked to be realizable.
Hence so too are all measurement configurations produced by the prescription. 
A 12-MZM example is  shown in Fig.~\ref{fig:Algorithm_2_nodes}F.

We thus find that implementing FPBC with our design could reduce the resource cost of MZM-based quantum computation.
The required number of MZMs is reduced, both since the computation is restricted to $\mathcal{R}_t$ and since no ancilla MZMs are needed.
We also reduce the total number of operations, by deleting all logical braids, and avoiding the overheads from braiding processes~\cite{Braiding}.

However, there is a residual limitation of locality in our design; it cannot be used for arbitrarily large registers $\mathcal{R}_t$.
In ideal systems, this arises via the suppression of $\omega_{\text{shift}}$ with the number $L$ of islands in the system. 
Since $E_J$ and $E_C$ characterize the effective JJ between bus- and ground-connected subsystems, they scale as $O(L)$ and $O(L^{-1})$, respectively, so we have $\delta\varepsilon_m\sim \exp(-cL)$ with $c$ a constant~\cite{E_J}.
$\omega_{\text{shift}}$ is further suppressed by a factor $A_{j,\alpha_j}/|\mathbf{A}_j|$ for every tri-junction around which one or two islands are bus-connected.
Hence, increasing the size of $\mathcal{R}_t$ requires the ability to resolve increasingly small $\omega_{\text{shift}}$.
In realistic setups, larger and more complex systems may also incur more accidental features (e.g., material defects, accidental quantum dots).
These may reduce coherence times and measurement fidelities~\cite{Defects_Nature2019,Trapped_QPs_transmons2020, Review_Decoherence_Supercond_Qubits2021, Materials_Challenge_QC2021}, and pose challenges for calibrating parity measurements.
However, one may be able to use techniques similar to those for mapping defect features and locations in transmon systems~\cite{Defects_Nature2019, Defect_location2020,TLS_Ensemble2021} to facilitate calibration, and reduce the number of defects with new materials techniques~\cite{Reducing_JJ_Capacitive_Loss2017, Coherence_Times_Transmon2021, Better_JJ_Fab2021, Bandaged_JJs2021, Materials_Challenge_QC2021}.
Additionally, in larger setups, more JJs allow for more quasi-particle poisoning events, which are not inhibited by a strong charging energy as they are in other designs~\cite{Karzig_Majorana2017}. 
However, these rates may still be small enough to be neglected on relevant timescales~\cite{MZM_Coherence_Times2018,Karzig_QPP_Maj_Qubits_2021}, and could be further reduced with quasi-particle traps~\cite{QP_Traps_2016,QP_Traps_2016_2}.

\textit{Conclusion:}
We have introduced fermion-parity-based computation, a low-resource-cost, measurement-based model of quantum computing with Majoranas, and have explained how it is able to simulate any fermionic quantum circuit.
We introduced a MZM hardware design that is free of constraints on measurable operators, beyond those of locality, and hence is well-suited to FPBC. 
We expect that a $t$-MZM FPBC, similarly to PBC~\cite{Bravyi_PBC2016}, can be simulated by a $(t-k)$-MZM FPBC if supplemented by exp($k$)-time classical processing; thus with FPBC one could minimize the quantum resources needed for fermionic computation. 
To overcome the locality constraint, future work could consider how multiple copies of our setup might be used to measure larger fermion parities. We expect one could adapt existing work on transmon qubit-parity measurements~\cite{Qubit_Parity_Mmt2010, Parity_Mmt_2_2010, Parity_Mmt_2013, Qubit_Parity_Mmt2018} to our Majorana-transmon setup, wherein frequency shifts are produced only by (suitably generalized~\cite{Akhmerov2010}) fermion parities.
%, since all transmon qubits remain in their ground states.
%
These larger setups could be made feasible by adapting transmon-based methods for improved measurements~\cite{Better_Transmon_Mmt_1_2017, Better_transmon_mmts2020} and large device design and calibration~\cite{2D_Cavity_grid_2009, Quantum_Supremacy2019, Calibrating_2_qubit_gate2019, Calibrating_devices_review2021, Suppressing_ZZ_int_scaling2021}.
One could also investigate our hardware design in the context of Majorana fermion codes~\cite{Maj_Ferm_Codes,QC_with_MFCs}, taking advantage of the large set of measurable operators.

We thank R. Jozsa for introducing us to PBC, and thank him and S. Strelchuk for useful discussions.
This work was supported by an EPSRC Studentship, EPSRC grant EP/S019324/1, and the ERC Starting Grant No. 678795 TopInSy.

\section*{Supplemental Material}

\renewcommand{\thesubsection}{\Roman{subsection}}
\renewcommand{\theequation}{S\arabic{equation}}
\renewcommand{\thefigure}{S\arabic{figure}}

\setcounter{equation}{0}
\setcounter{figure}{0}

\subsection{Logical braids exhaust all parity-preserving Clifford-like unitaries}

The Clifford group, defined for qubits, is the group of unitary operators that send strings of Pauli operators to other such strings under conjugation~\cite{Gottesman_thesis}.
Here, in analogy with the Clifford group, we demonstrate that $W_4$ operators can generate all parity-preserving unitaries that send Majorana strings to Majorana strings.
We denote the group of all such Clifford-like unitaries $\mathcal{B}_{2n}$ for a system of Majoranas $\gamma_1,\ldots,\gamma_{2n}$.
Suppose $W$ is some basic operator satisfying $W\gamma_{2n}W^\dagger = \Upsilon$ for $\Upsilon\in \text{Maj}(2n)$. Then we can obtain any other operator $U$ satisfying the same property through $U = W U^\prime$, where $U^\prime = W^\dagger U$ satisfies $U^\prime \gamma_{2n} (U^\prime)^\dagger = \gamma_{2n}$. We can see that $U^\prime \in \mathcal{B}_{2n-1}$, since any $U^\prime \gamma_i (U^\prime)^\dagger$, $i\neq 2n$, must be parity-odd and anti-commute with $\gamma_{2n}$, and hence cannot feature $\gamma_{2n}$.
(While considering $U = V W U^\prime$ with $V\Upsilon V^\dagger = \Upsilon$ appears more general, any $V\in\mathcal{B}_{2n}$ satisfying $V\Upsilon V^\dagger = \Upsilon$ is such that $V W = W V^\prime$ for some $V^\prime \in\mathcal{B}_{2n-1}$.)
This factorization is analogous to that for the Clifford group~\cite{nielsen_chuang}, and as in that case, it allows one to build $U\in\mathcal{B}_{2n}$ iteratively from a sequence of basic operators.

We now demonstrate that $W_4$ is a suitable basic operator. 
If $\Upsilon$ has no $\gamma_{2n}$ factor, $W=\exp(\frac{\pi}{4}\Upsilon\gamma_{2n})$ satisfies the required property, $W\gamma_{2n}W^\dagger = \Upsilon$.
Otherwise take some $\gamma_j$ absent from $\Upsilon$; $\Upsilon^\prime =\exp(\frac{\pi}{4}\gamma_{2n}\gamma_{j})\Upsilon \exp(-\frac{\pi}{4}\gamma_{2n}\gamma_{j})$ contains no factor of $\gamma_{2n}$.
Hence the operator $W=\exp(-\frac{\pi}{4}\gamma_{2n}\gamma_{j})\exp(\frac{\pi}{4}\Upsilon^{\prime}\gamma_{2n})$ is of the desired type.
Since $W_4$ can send $W_{2m}\mapsto W_{2m\pm 2}$ under conjugation, for any $1<m<n-1$, the $W_4$ can generate the above basic operators whenever $n>2$.
They can also map between various $W_{2m}$ operators, since, for example, $\exp(\frac{\pi}{4}\gamma_a \Upsilon) = U \exp(\frac{\pi}{4}\gamma_b\Upsilon) U^\dagger$ where
\begin{align}
    U = \exp\left(\frac{\pi}{4}\gamma_a\Upsilon^\prime\right) \exp\left(\frac{\pi}{4}\Upsilon^\prime \gamma_b\right)
\end{align}
for some 3-Majorana $\Upsilon^\prime$ that has even overlap with $\Upsilon$.

This shows that all elements of $\mathcal{B}_{2n}$ are products of $W_{4}$ operators.
Hence logical braids exhaust $\mathcal{B}_{2n}$: the group $\mathcal{B}_{2n}$ is the logical braid group.

\subsection{Measurements of all fermion parities can be generated by FPBC}
\label{app:gadgets}

Here, we show that, for dense encodings and for $t\leq n$, FPBC can result in arbitrary fermion parities in the measurement sequence.
This is true for any magic state gadget one might utilize, provided that it uses one magic state to implement an $\exp(\theta\gamma_a\gamma_b)$ gate, or more generally $\exp(i\theta \Gamma)$ with fermion parity $\Gamma$. 
(We comment on sparse encodings and other more general cases at the end of this section.)
Thus, if MZMs are hosted in a setup which has configuration constraints on measurable operators, it is possible for FPBC to require the measurement of unmeasurable operators; in this case, braids or other operations must supplement FPBC.

For the fermionic circuit being simulated, we have two Majorana registers, each in a fixed parity sector. 
A computational register $\mathcal{R}_n$ of $2n+2$ Majoranas is in an arbitrary state, while the register $\mathcal{R}_t$, with $2t+2$ Majoranas, is in a state encoding $t$ magic states (for concreteness, let this state be the $\ket{\psi^{(t)}}$ defined in the main text).
We consider a gadget $\mathcal{M}_j$ enacting $\exp(i\theta \Gamma_j)\notin \mathcal{B}_{2n+2}$ (with $\Gamma_j=-i\gamma_{a_j}\gamma_{b_j}$)
on $\mathcal{R}_n$ while using the $j$\textsuperscript{th} magic state.
[More general $\Gamma_j$ are achieved upon suitably altering the purely $\mathcal{R}_n$ logical braids $B_i$ described below.]
Apart from some basic assumptions set out below, our considerations hold regardless of the specific form of $\mathcal{M}_j$.
As in the main text, we define $s_l = i\gamma_{2l-1}\gamma_{2l}$ and $X_l$ such that $[X_l, X_{l'}]=[X_l,s_{l'}]=0$ for $l\neq l'$ and $\{X_l,s_l\}=0$, $l\in\{1,2,\ldots,t\}$.
Since the $s_l$ eigenvalues label a complete eigenbasis for $\mathcal{R}_t$ with definite $\Gamma_{2t+2}$ eigenvalue, and the $X_l$ flip these eigenvalues, the operators $s_l$, $X_l$ and $\Gamma_{2t+2}$ generate  $\overline{\text{Maj}}(2t+2)$. 

\begin{figure}[t]
    \centering
    \includegraphics[width=0.99\linewidth]{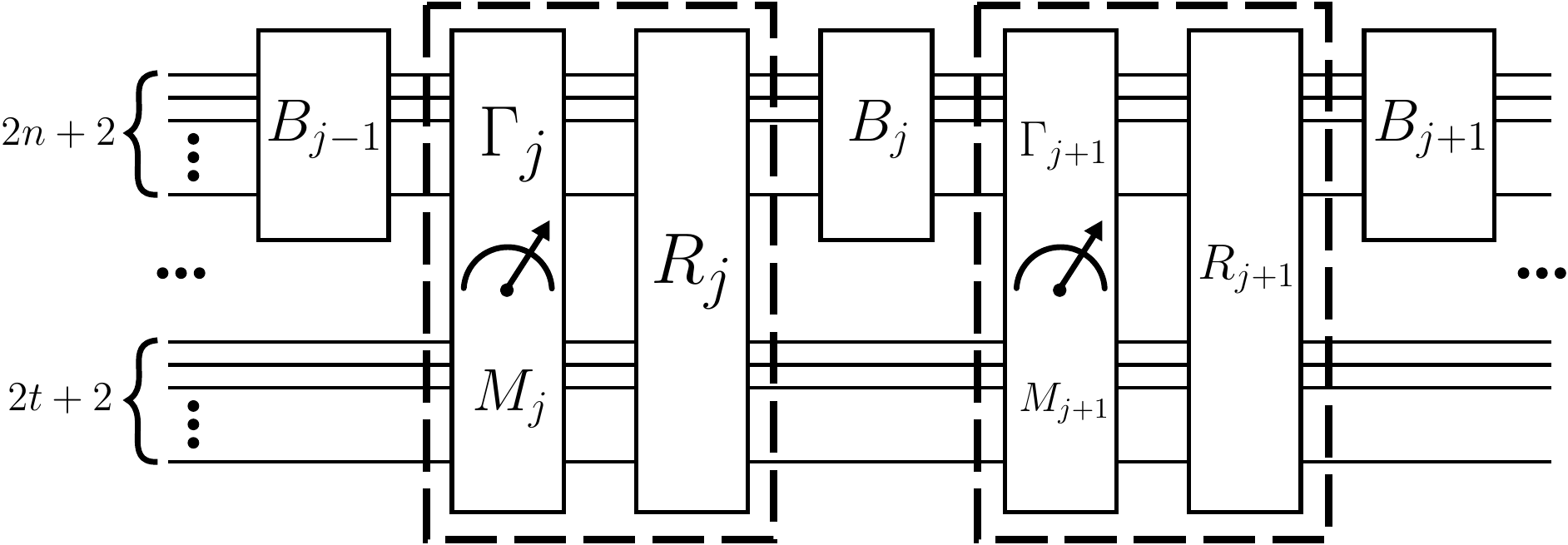}
    \caption{Section of a generic fermionic circuit. The gadget (dashed boxes) labelled $i$ (for $i=j,\, j+1$) uses the $i$\textsuperscript{th} magic state. It involves a two-register measurement followed by a logical braid $R_i$. The gadgets are all interspersed with logical braids $B_i$ acting on the computational register.}
    \label{fig:Ferm_Circuit}
\end{figure}

We assume that $\mathcal{M}_j$ is composed of fermion parity measurements and logical braids, and acts trivially on $\gamma_c, \, c\neq a_j,b_j$
in $\mathcal{R}_n$, and on magic states $j'\neq j$ in $\mathcal{R}_t$.
We use the convention in Fig.~\ref{fig:Ferm_Circuit}: the gadget begins with measurements and finishes with logical braids.
(This is equivalent to other $\mathcal{M}_j$ conventions upon commuting all logical braids in $\mathcal{M}_j$ past its measurements.)
Then, the gadget's measurements cannot be trivial on either register: that would either destroy some logical information or collapse a magic state.
Therefore,  these measurements are all in the form $\mathcal{P}_j=\Gamma_j M_j$ for some operator $\openone \neq M_j\in\langle X_j, s_j\rangle$, where $\langle X_j, s_j\rangle$ is the set of operators generated by $s_j,X_j$ and the phase factor $i$.
There need  be only one such measurement, since any subsequent one must either have certain measurement outcome, in which case it is deleted and the outcome computed, or anti-commute with a preceding measurement; in this case we replace it with a logical braid, as described in the main text.
For concreteness, we assume that these measurements, together with those of $\{s^{\text{(c)}}_a\}$ in $\mathcal{R}_n$ at the beginning and the end of the computation (as described in the main text), are the only measurements in the circuit.

Since the gadget serves only to enact the gate $\exp(i \theta \Gamma_j)$, it must leave the two registers disentangled after its use.  
Hence $R_j$ in Fig.~\ref{fig:Ferm_Circuit} (which is adaptive, i.e., could depend on the gadget's measurement outcome) must disentangle the two registers after the measurement of $\Gamma_j M_j$. 
This implies that $R_j$ neither commutes with $\Gamma_j M_j$ nor acts trivially in either register, nor is a product solely of factors that do.
Therefore, upon absorbing any purely $\mathcal{R}_n$ logical braids from $R_j$ into $B_{j+1}$ (thus making $B_{j+1}$ adaptive) and commuting any purely $\mathcal{R}_t$ logical braids in $R_j$ past the circuit (these commute with the $M_{j'\neq j}$ and $R_{j'\neq j}$ by our assumption that the gadget acts trivially on magic states $j'\neq j$), we are left with  $R_j=\exp(i\frac{\pi}{4}\Gamma_j N_j)$ for some Hermitian $N_j\in\langle X_j, s_j\rangle$ satisfying $\lbrace N_j, M_j\rbrace = 0$ and possibly depending on the gadget's measurement outcome.
Since $N_j$ and $M_j$ anti-commute, $\langle N_j, M_j\rangle = \langle X_j, s_j\rangle$.

We now show that FPBC can require arbitrary fermion parity measurements.
We first commute all the $B_i$ past the circuit. 
For the $j$\textsuperscript{th} gadget, this replaces $\Gamma_j$ by a parity operator $\mathcal{O}_j$, respectively updating $\mathcal{P}_j$ and $R_j$ to
\begin{equation}
\mathcal{P}_j= \mathcal{O}_j M_{j},\qquad R_j=\exp(i\frac{\pi}{4}\mathcal{O}_j N_j).
\end{equation}
(We leave the symbols for $\mathcal{P}_j$ and $R_j$ unchanged through these updates until later on, when it is convenient to define new ones.) 
The final measurements in $\mathcal{R}_n$ at the end of the circuit are updated to $\Lambda_a$ (with $a= 1,\ldots, n$). 
Choosing a fermionic circuit to be simulated and a set of gadget measurement outcomes determines the $B_i$, so that each $B_i$ can be set to be any member of $\mathcal{B}_{2n+2}$. 
Thus, the $\mathcal{O}_j$ can be set to be any arbitrary fermion parities (apart from $\openone$ or $\Gamma_{2t+2}$), given a suitable choice of circuit and of gadget measurement outcomes.

For the following, let $\mathcal{S}=\langle \Gamma_{2n+2}, s_1^\text{(c)},\ldots, s_n^\text{(c)}\rangle$ where $\Gamma_{2n+2}$ is $\mathcal{R}_n$'s overall fermion parity. 
We make the assumption that $n\geq t$. 
Since we can choose the $\mathcal{O}_i$ arbitrarily and there are only $t$ of these operators, let us choose them to be such that no non-trivial product of the $\mathcal{O}_i$ is a member of $\mathcal{S}$. That is, the $\mathcal{O}_i$ are ``independent" from $\mathcal{S}$.
For fermion parity operator $O$, define $\mathcal{A}(O)$ to be the set of all $s_a^\text{(c)}$ anti-commuting with $O$.
Since $\{s_a^\text{(c)}\}$ and $\Gamma_{2n+2}$ label a complete basis in $\mathcal{R}_n$, the group $\mathcal{S}$ is maximal:
$\mathcal{R}_n$ has no parity-preserving operator outside of $\mathcal{S}$ that commutes with $\mathcal{S}$.
This, along with the independence of the $\mathcal{O}_i$ from $\mathcal{S}$, implies that if $O$ is a non-trivial product of the $\mathcal{O}_i$, then $\mathcal{A}(O)\neq \emptyset$.

We next study the consequences of commuting the $R_j$ past the $\mathcal{P}_j$ measurements.
For the resulting updates, we shall use the following, which can be easily verified:
\begin{lemma}\label{lem:braidABC}
Consider fermion parities $A$, $B$, $C$ satisfying $AB = (-1)^{a} BA$, for $a\in\lbrace 0,1\rbrace $ and $\lbrace AB, C\rbrace = 0$.
Then conjugation with the logical braid $U_{AB}=\exp(i^{1-a}\frac{\pi}{4}AB)= \frac{1}{\sqrt{2}}(1 + i^{1-a} AB)$ gives
\begin{equation}
U_{AB}^\dagger CU_{AB} = i^{1-a} CAB.
\end{equation}
\end{lemma} 

We commute $R_{t-1}$ past $\mathcal{P}_t$, then $R_{t-2}$ past $\mathcal{P}_{t-1}$ and $\mathcal{P}_t$ etc; in this way the $R_j$ need not be commuted past other logical braids. 
While we always have $[N_{t-1},M_t]=0$, we can either have $[\mathcal{O}_{t-1},\mathcal{O}_{t}]=0$ or $\{\mathcal{O}_{t-1},\mathcal{O}_{t}\}=0$ 
(either case can arise for some $B_{t-1}$).
In the former case, commuting $R_{t-1}$ past $\mathcal{P}_t$ leaves $\mathcal{P}_t$ unchanged.
In the latter case, by Lemma~\ref{lem:braidABC},  it updates 
\begin{equation}\label{eqn:Mmt_Op_Gadget}
\mathcal{P}_t=\mathcal{O}_{t} M_{t} \mapsto \mathcal{P}_t=i\mathcal{O}_{t} \mathcal{O}_{t-1} N_{t-1} M_{t}.
\end{equation}
In either case, $\mathcal{P}_{t-1}$ and the resulting $\mathcal{P}_t$ commute (in the latter case due to $\{N_{t-1},M_{t-1}\}=\{\mathcal{O}_{t-1},\mathcal{O}_{t}\}=0$). 
Similarly, after commuting through all $R_j$ between $\mathcal{P}_i$ and $\mathcal{P}_l$ ($i<l$), the updated $\mathcal{P}_l$ commutes with $\mathcal{P}_i$.
The remaining $R_{j<i}$ do not change this, because they either update $\mathcal{P}_i$ and $\mathcal{P}_l$ in the same way, or they update one with a factor with which the other commutes.
Thus the resulting $\mathcal{P}_i$ all mutually commute.
However, they might not all commute with the $\{s^{\text{(c)}}_a\}$ at the beginning.

We now use our remaining freedom in choosing $\mathcal{O}_i$ to judiciously set the resulting $\mathcal{P}_i$.
Since each has a different, non-trivial product of $\mathcal{O}_i$ (already set to be independent of $\mathcal{S}$), the $\mathcal{P}_i$ are independent of $\mathcal{S}$, and uniquely set by choosing the $\mathcal{O}_i$. 
Thus, so long as $\mathcal{A}(P)\neq \emptyset$ for $P$ any non-trivial product of the $\mathcal{P}_i$, we may choose the $\mathcal{A}(\mathcal{P}_i)$ arbitrarily.
Let us, therefore, take $\mathcal{A}(\mathcal{P}_i) = \{s_{a_i}^\text{(c)}\}$ for some indices $a_i\in \{1,\ldots, n\}$, with $a_i\neq a_j$ for $i\neq j$.
This choice clearly ensures the $\mathcal{P}_i$ are independent of $\mathcal{S}$.

Beginning with $\mathcal{P}_1$, this measurement anti-commutes only with $s_{a_1}^\text{(c)}$. 
Hence it is replaced with the logical braid $V_1 = \exp(\frac{\pi}{4} \mathcal{P}_1 s_{a_1}^\text{(c)})$.
Since both $\mathcal{P}_1$ and $s_{a_1}^\text{(c)}$ commute with all $\mathcal{P}_{i>1}$, $V_1$ can be commuted past these measurements without updating them. 
Then we must replace $\mathcal{P}_2$ with a logical braid, $V_2 = \exp(\frac{\pi}{4}\mathcal{P}_2 s_{a_2}^\text{(c)})$, which can be similarly commuted past all $\mathcal{P}_{i>2}$, etc.
After replacing all measurements of the $\mathcal{P}_i$ with their corresponding $V_i$, we are left with a circuit involving measurements of the $s_a^\text{(c)}$, followed by the logical braid $(V_t V_{t-1}\ldots V_1)(R_1 R_2 \ldots R_t)$, which is followed by the measurements of the $\Lambda_a$.

Owing to the arbitrariness of the logical braid $B_t$ that updated only the final $s_a^\text{(c)}$ measurements but not the $\mathcal{O}_i$, we may choose the $\Lambda_a$ to be any (non-trivial, mutually commuting) fermion parities. Let us focus on the first of these, $\Lambda_1$. 
It is possible for this to have any commutation relations with the $s_a^\text{(c)}$ and $\mathcal{O}_i$ (so long as it does not commute with all of them), because there are $t\leq n$ of the $\mathcal{O}_i$ and $n$ of the $s_a^\text{(c)}$, all of which are independent.
Let $\Lambda_1 \mathcal{O}_i = (-1)^{q_i} \mathcal{O}_i\Lambda_1$ for $q_i\in \lbrace 0, 1\rbrace$. 
Then the above implies we may choose all $q_i$ and the set $\mathcal{A}(\Lambda_1)$ arbitrarily.

Commute all logical braids $R_t,\ldots, R_1$ past $\Lambda_1$. 
This updates $\Lambda_1$ to $\mathcal{Q}_1$, an operator commuting with all $\mathcal{P}_i$ for the same reason that the $\mathcal{P}_i$ mutually commute.
Now we show that $\mathcal{Q}_1$ may contain any string of $N_i$:
if $I$ is some set of indices from $\lbrace 1,\ldots, t\rbrace$, then $\mathcal{Q}_1$ can act in $\mathcal{R}_t$ as $\prod_{i\in I}N_i$  if we choose the $q_i$ such that $\{ \Lambda_1 (\prod_{j>i; \, j\in I} \mathcal{O}_j), \mathcal{O}_i\} = 0$ if and only if $i\in I$ (Lemma~\ref{lem:braidABC}).
Thus, given operators $\mathcal{O}_i$, there is a one-to-one correspondence between length-$t$ bit strings $\mathbf{q}=(q_1,\ldots, q_t)$ and products of the $N_i$ and hence, we can choose $\mathcal{Q}_1$ to include any product of $N_i$.

Now we consider commuting the $V_i$ past $\mathcal{Q}_1$. 
Since $\mathcal{Q}_1$ and the $\mathcal{P}_i$ commute, $V_i$ updates $\mathcal{Q}_1$ only if $\{ \mathcal{Q}_1,s_{a_i}^\text{(c)}\} = 0$. 
Thus, by Lemma~\ref{lem:braidABC}, the $V_i$ update $\mathcal{Q}_1$ to the operator
\begin{equation}
    \mathcal{Q}_1^\prime = V_t^\dagger \ldots V_1^\dagger \mathcal{Q}_1 V_1 \ldots V_t = \mathcal{Q}_1 \left(\prod_{s_{a_j}^\text{(c)}\in \mathcal{A}(\mathcal{Q}_1)} \mathcal{P}_j s_{a_j}^\text{(c)} \right).
\end{equation}
Each $\mathcal{P}_j$ acts as $M_j$ on the $j$\textsuperscript{th} magic state, trivially on all $k>j$ magic states and either trivially or as $N_k$ on all $k<j$ magic states. 
We may thus obtain an arbitrary string of $M_j$ operators in $\mathcal{Q}_1^\prime$ if we can set $\mathcal{A}(\mathcal{Q}_1)$ arbitrarily. 
But we know that we may set $\mathcal{A}(\Lambda_1)$ arbitrarily, 
and that $\mathcal{Q}_1  = \pm \Lambda_1 \prod_{j\in I} (i\mathcal{O}_j N_j)$ for some set of indices $I$ (where the $\pm 1$ is due to $I$ being unordered). 
Hence if $s_a^\text{(c)} \in \mathcal{A}(\prod_{j\in I}\mathcal{O}_j)$ then $s_a^\text{(c)} \in \mathcal{A}(\mathcal{Q}_1)$ if and only if $s_a^\text{(c)}\notin \mathcal{A}(\Lambda_1)$
and vice versa.
Thus for any product of $\mathcal{O}_i$ in $\mathcal{Q}_1$, we can  choose $\mathcal{A}(\mathcal{Q}_1)$ arbitrarily by choosing a suitable $\mathcal{A}(\Lambda_1)$.

To summarise, by choosing a suitable $\mathbf{q}$, the operator $\mathcal{Q}_1^{\prime}$ can include any string of $N_i$, and by choosing a suitable $\mathcal{A}(\Lambda_1)$, it can include any string of $M_i$.
(Of course, the $\mathcal{P}_j$ factors in $\mathcal{Q}_1^\prime$ may also include some $N_i$, but these are independent of $\mathbf{q}$, and hence do not preclude using $\mathbf{q}$ to achieve any string of $N_i$ in $\mathcal{Q}_1^\prime$.)
Thus we can choose $\mathcal{Q}_1^\prime$ to contain any string from $\langle X_1, s_1, \ldots, X_t, s_t\rangle$. 

What remains to show is that the restriction of $\mathcal{Q}_1^\prime$ to $\mathcal{R}_t$ may need to be measured in FPBC (instead of being replaced by a logical braid). 
For this, $\mathcal{Q}_1^\prime$ must commute with all the $s_{a}^\text{(c)}$. 
To see that this is possible, note that having set $\mathcal{A}(\mathcal{P}_i) = \{s_{a_i}^\text{(c)}\}$ implies that $\mathcal{Q}_1^\prime$ commutes with all $s_a^\text{(c)}\notin \mathcal{A}(\mathcal{Q}_1)$.
Furthermore, $\mathcal{Q}_1^\prime$ commutes with $s_{a_j}^\text{(c)}\in\mathcal{A}(\mathcal{Q}_1)$ because for such $s_{a_j}^\text{(c)}$, we have $\{\mathcal{Q}_1,s_{a_j}^\text{(c)}\}=\{\mathcal{P}_j,s_{a_j}^\text{(c)}\}=0$ and $[\mathcal{P}_i,s_{a_j}^\text{(c)}]=0$ for $i\neq j$. 
This leaves only $s_a^\text{(c)} \in \mathcal{A}(\mathcal{Q}_1)\setminus S_t$, with $S_t=\{s_{a_j}^\text{(c)}|j=1,\ldots,t\}$, to be considered. 
However, thus far only $\mathcal{A}(\mathcal{Q}_1)\cap S_t$ needed to be chosen, hence we can set $\mathcal{A}(\mathcal{Q}_1)\setminus S_t=\emptyset$. 
Therefore, with the above choices, $[\mathcal{Q}_1^\prime,s_{a}^\text{(c)}]=0$ for all $a=1,\ldots,n$. 
We have thus proved the following:

\begin{theorem}
Let $\mathcal{R}_n$ be a register of $(2n+2)$ Majoranas in a fixed parity sector, and in the $s_a^\text{(c)} = 1 \, (a=1,\ldots, n)$ computational basis state. 
Let $T_{\theta,\, ab} = \exp(\theta\gamma_a\gamma_b) \notin \mathcal{B}_{2n+2}$ be some gate implementable via a magic state gadget that uses a single magic state.
Consider the class of fermionic circuits defined on $\mathcal{R}_n$ involving $t\leq n$ gates of the form $T_{\theta,\, ab}$ interspersed by logical braids $B_i\in \mathcal{B}_{2n+2}$, and $n$ final $s_a^\text{(c)}$ measurements. 
Upon converting such a circuit to FPBC, the set of all possible measurements thus generated may contain any fermion parity operator defined on a register $\mathcal{R}_t$ of $(2t+2)$ Majoranas.
\end{theorem}

The above theorem implies that if a MZM design contains configuration constraints on measurable operators, it is possible for FPBC to require the measurement of unmeasurable operators (cf. Section~\ref{app:resourcecostconstrained}).
While in the most general case, without the stated assumptions or judicious choices, FPBC could avoid certain measurements, it remains difficult for it to avoid unmeasurable (or uneasily measurable) operators in many designs.
For example, for $t>n$, while some strings from $\langle X_1, s_1, \ldots, X_t, s_t\rangle$ may not require measurement, those that do may still involve any of the $s_i$ and $X_i$. 
For sparse encodings, we must consider a subset of the above-defined $s_j$ and $X_j$ as encoding logical information; 
then any string of operators from the group generated by this subset can appear in FPBC. 
Similarly, for more general gadgets using $k\geq 1$ magic states to enact a single gate, while FPBC can generate measurements drawn only from a subset of all MZM strings, these strings have similar complexity as those for $k=1$.

\subsection{Resource cost of FPBC in constrained hardware}
\label{app:resourcecostconstrained}

Section~\ref{app:gadgets} showed that it is possible to generate an arbitrary parity-preserving MZM string (up to a factor of $\Gamma_{2t+2}$) in FPBC.
Hence let us model FPBC as a sequence of $p\leq t$ (commuting) fermion parity measurements chosen uniformly at random. 
(We again assume a dense encoding. Some designs are not suitable for this~\cite{Majorana_Ram_Flux_control2013,Karzig_Majorana2017}; we briefly comment on them at the end of this section.)
In terms of this model, we now show that, upon increasing the number of MZMs, FPBC becomes overwhelmingly likely to generate many measurements that, in a broad class of physical setups, require performing prohibitively many extra operations, such as braids.

Generalising from transmon-based designs, we consider setups that have readout islands $v_i$ hosting a constant, even number $c$ of MZMs. 
We assume the islands are arranged in a cubic lattice $\mathcal{L}$ of dimension $d$.
For a fermion parity $\Gamma$, we define the weight of $\Gamma$ in $v_i$ to be the number of MZMs in $v_i$ that feature in $\Gamma$.
Motivated by features common to many proposed designs~\cite{Maj_Based_FQC2018,Coulomb_assisted_braiding2012,Milestones_Paper2016}, we assume that (i) for some $\Gamma$ to be measurable, it must have even weight in all islands of $\mathcal{L}$,
and (ii) there is a constant braid radius $r$, such that every MZM within distance $r$ of some $\gamma_a$ can be braided directly with $\gamma_a$ (either adiabatically or through measurements with ancillas).

As noted in the main text, if an operator $\Gamma$ is non-measurable, braids (or extra measurements and ancillas) must supplement FPBC.
We define the braid cost of parity $M_i$ to be the minimum number of braids that must be performed to place $M_i$ in a measurable configuration. 
Every operator that has odd weight in two readout islands has a braid cost given by $\sim x/r$, where $x$ is the distance (i.e. the number of lattice links) between the two islands. 
So, in general, FPBC might involve the measurement of $O(m)$-braid-cost operators, where $m=O(t^{1/d})$ is the linear size of the system.

Below, we upper bound the proportion of parity operators that have braid cost no greater than a given value $R$. 
An operator with braid cost less than $R$ will be called $R$-measurable.
We impose periodic boundary conditions on $\mathcal{L}$ (this only increases the proportion of measurable operators).
As we focus only on asymptotics, assume for simplicity that $\mathcal{L}$ has equal side lengths $m$, and that $m$ is divisible by $rR$.
Divide the lattice into $d$-dimensional cubic regions of side length $rR$. 
There are $(rR)^d$ inequivalent such divisions. 
We now prove the following lemma:
\begin{lemma}\label{lem:R-meas_ops}
Given an $R$-measurable operator $\Gamma$, lattice $\mathcal{L}$ (with periodic boundary conditions) can be divided into $d$-cubic regions of side length $rR$ in such a way that $\Gamma$ has even weight in every region.
\end{lemma} 
\begin{proof}
Suppose $\Gamma$ has odd weight in $2q$ islands with distances $b_1,\ldots,b_q$ between nearest-neighbor islands.
Since $\Gamma$ is $R$-measurable, we have $B\coloneqq\sum_i b_i < rR$.
Start in $d=1$. 
There are precisely $b_j$ divisions of $\mathcal{L}$ for which a boundary between regions lies between nearest-neighbor pair $j$ (with separation $b_j$). 
We say that pair $j$ is ``split" by those divisions.
We now show that a division exists that splits none of the pairs. 
A spatial distribution of the $2q$ islands that maximises the number of pair-splitting divisions is one where, for any division, at most one of the pairs is split.
In this case there are $B$ divisions that split some pair. 
But since $B<rR$, even in this case there must be at least one division that splits no pairs.
For $d>1$, the sum of separations between nearest-neighbor pairs in any direction is at most $B< rR$. 
Hence the argument for $d=1$ can be applied to all directions.
\end{proof}

Next we provide an upper bound on the number $M_{m,d}^R$ of $R$-measurable fermion parities (up to a sign) in $\mathcal{L}$. 
While for any $R$-measurable parity $\Gamma$ there exists at least one division such that $\Gamma$ has even weight in every region (Lemma~\ref{lem:R-meas_ops}), not all operators that satisfy this are $R$-measurable. 
[E.g., an operator with a pair (in terms of Lemma~\ref{lem:R-meas_ops}) in every region is not $R$-measurable beyond a certain $m$ if $R< O(m)$, since we have $O(m)$ regions.]
Hence, by counting all parities that have even weight in every region for at least one division, we over-estimate $M_{m,d}^R$. 
For a given division, there are $(m/rR)^d$ regions with $c(rR)^d$ MZMs in each (recall $c$ is the number of MZMs per island).
The number of parities up to a sign in a single region is $2^{c(rR)^d-1}$ (where the $-1$ accounts for half of all operators on the region being parity odd).
Hence there are $2^{m^d c - (m/rR)^d}$ operators that have even weight in every region.
Multiplying this value by the number of distinct divisions $(rR)^d$ further over-estimates $M^R_{m,d}$, since operators that are even in all regions for multiple divisions are counted more than once. 
There are $N_{m,d} = 2^{m^d c - 1}$ parity operators in total (up to a sign). 
We conclude:
\begin{align}\label{eqn:R_meas_ops_dim_d}
    \frac{M_{m,d}^R}{N_{m,d}} < (rR)^d\, 2^{1- (m/rR)^d}.
\end{align}
Hence, we have proved the following:
\begin{theorem}\label{thm:Braid_Cost}
Let $\mathcal{L}$ be a $d$-dimensional cubic lattice with $m^d$ islands and a constant braid radius $r$. For any $R<O(m)$, the proportion of $R$-measurable operators $M_{m,d}^R/N_{m,d}\rightarrow 0$ as $m\rightarrow\infty$.
\end{theorem} 
A random set of fermion parity measurements will therefore be dominated by non-$R$-measurable operators for any $R<O(m)$, as $m$ increases.
Furthermore, this $O(m)$ braid cost cannot be avoided by introducing ancilla MZMs into the system.
For example, if we assume that having an even weight in all islands is the only requirement for measurability, then an operator $i\gamma_a\gamma_b$ that is non-$R$-measurable can be made $R$-measurable only by using an ancilla operator $A$ that has odd weight within regions of radius $rR/2$ centred on the islands containing $\gamma_a$ and $\gamma_b$ (call these islands $v_a$ and $v_b$ respectively).
One could then implement $<R$ braids to move $\gamma_a$ and $\gamma_b$ onto the closest islands in which $A$ has odd weight (call these $v_{\tilde{a}}$ and $v_{\tilde{b}}$ respectively).
Then the product $i\gamma_a\gamma_b A$ is measurable and, since the eigenvalue of $A$ is known, that of $i\gamma_a\gamma_b$ can be inferred.
But if $v_a$ and $v_b$ are separated by a distance of more than $2Rr$ [a scenario that according to Theorem~\ref{thm:Braid_Cost} is overwhelmingly likely as $m\rightarrow \infty$, for any $R<O(m)$],
the operator $A$ must have odd weight in islands that are more than a distance of $Rr$ separated. 
Hence, $A$ itself is not $R$-measurable and so the system cannot be prepared in a state of definite $A$-parity with fewer than $R$ braids.
Indeed, the number of required braids is no different than if $A$ had not been introduced: if $v_{\tilde{a}}$ and $v_{\tilde{b}}$ are separated by distance $rD$, there are $D$ braids required to prepare the system in a fixed $A$-parity state, and $\sim R$ braids to move $\gamma_a$ to $v_{\tilde{a}}$ and $\gamma_b$ to $v_{\tilde{b}}$. 
But $r(R+D)$ is at least as big as the shortest distance between $v_a$ and $v_b$, and hence introducing $A$ cannot reduce the number of braids required to measure $i\gamma_a\gamma_b$.

Finally, consider the case in which MZMs are grouped into smaller blocks of fixed parity (e.g., in ``Majorana RAM"~\cite{Majorana_Ram_Flux_control2013} and tetron/hexon~\cite{Karzig_Majorana2017,Maj_Box_Qubits2017,Quantum_Dot_Readout2020,Smith_Bartlett_Readout2020} architectures).
These setups cannot implement fermionic quantum circuits~\cite{Bravyi_Kitaev_Fermionic_QC2002,Maj_Based_FQC2018}, but may still perform FPBC with magic states sparsely encoded. 
There may still be, however, a large number of non-measurable operators generated in FPBC if, for example, the system constrains measurable operators to have support in blocks that are geometrically grouped close together.
(This is the case for tetron and hexon architectures~\cite{Karzig_Majorana2017,Optimizing_Clifford_Gates_2020}, but not for the Majorana RAM~\cite{Majorana_Ram_Flux_control2013}.)
As we noted in Section~\ref{app:gadgets}, this FPBC could contain arbitrary measurements, subject to the constraint that all measurement operators commute with all block parities.
We expect (but do not prove here), similarly to the number of $R$-measurable [$R< O(m)$] operators in systems using dense encodings, that the proportion of operators that are measurable, or can be brought into a measurable configuration with only few extra operations, becomes vanishingly small as the size of the system increases.
While Majorana RAMs do not have this constraint, in that case braids, moving computational and ancilla MZMs onto/off the bus-connected islands, are required to switch between different Pauli operator strings. 
In the worst case, one must perform $O(t)$ braids per measurement (if two subsequent Pauli strings differ greatly from one another) and $O(t)$ measurements - hence one might require $O(t^2)$ braids to perform FPBC.

\subsection{Top-transmon Hamiltonian, tri-junction projections, and resonance shifts}

Here we discuss the Hamiltonian of the proposed physical design (cf. Fig.~\ref{fig:Model} of the main text). 
Let islands belonging to set $S_\mathcal{B}$ be bus-connected and the rest be ground-connected. 
Furthermore, assume $E_J\gg E_C$ and that $E_{J,ak}^{\text{(on)}}/E_{C,k}\gg E_J/E_C$ for all superconductors $a,k$ connected by JJs.
Thus all bus-connected (ground-connected) islands are taken to have their superconducting phases pinned to that of the bus (phase ground).
We neglect the effects of quantum phase slips of these island phases around their minima, taking into account only the quantum phase slips of the bus-ground phase difference.
In this case, and for decoupled tri-junctions [i.e., in the absence of $V_{M,j}$ of the main text, repeated in Eq.\eqref{eq:apptrijun} for convenience], the system's energy levels are approximately given by~\cite{Transmon_Qubit2007,Majorana_Ram_Flux_control2013}
\begin{multline}\label{eqn:transmon_energies}
    E_m \approx \bar{E}_m 
    -(-1)^m \delta\varepsilon_m \\ 
    \left(\prod_{k\in S_\mathcal{B}} \mathcal{P}_k \right) \cos\left[\frac{\pi}{e}\left(q_0+\sum_{k\in S_\mathcal{B}} q_k\right)\right],
\end{multline}
where $\mathcal{P}_k$ is the fermion parity of the Majorana bound states on island $k$, 
$q_0$ and $q_k$ are related to the induced charges on the bus and island $k$, respectively~\cite{Top_Transmon2011},
and the form of $\bar{E}_m$ and $\delta\varepsilon_m$ can be found in Ref.~\onlinecite{Transmon_Qubit2007}. 
The relevant energy scales are $\hbar\Omega_m =\bar{E}_{m+1}-\bar{E}_m \approx \sqrt{8E_JE_C}$ and \mbox{$\delta\varepsilon_m \propto \exp(-\sqrt{8E_J/E_C})\ll \hbar\Omega_m$.} 

We next include $V_{M,j}$. 
In the $E_M\ll \hbar\Omega_m$ regime, the $V_{M,j}$ term acts as a perturbation to the transmon levels. 
If the temperature $T$ also satisfies $k_BT\ll \hbar\Omega_0$, one can focus on the two lowest transmon levels ($m=0,1$), described by the perturbed top-transmon Hamiltonian
\begin{multline}
    H_\text{tt} = \sigma_z \left[\tfrac{1}{2}\hbar\Omega_0 + \left(\prod_{k\in S_\mathcal{B}}\mathcal{P}_k\right) \delta_+ \cos(\pi q_{\text{tot}}/e)\right]\\
    +\left(\prod_{k\in S_\mathcal{B}}\mathcal{P}_k\right)\delta_- \cos(\pi q_{\text{tot}}/e) + \sum_j V_{M,j},
\end{multline}
where $\delta_{\pm} = \tfrac{1}{2}(\delta\varepsilon_1\pm \delta\varepsilon_0)$ and $q_{\text{tot}} = q_0 + \sum_{k\in S_\mathcal{B}} q_k$. 
The Pauli matrix $\sigma_z$ acts on transmon eigenstates labelled by $m=0,1$.
From now on, we assume $q_\text{tot} = 0$, for simplicity.

We now focus on the tri-junction terms
\begin{equation}\label{eq:apptrijun}
V_{M,j}=
\frac{1}{2}E_{M}\sum_{abc=1}^{3}\epsilon_{abc}A_{j,a}(i\gamma_{j,b}\gamma_{j,c}).
\end{equation}
Introducing the vectors
\begin{equation}
\mathbf{A}_{j}=\left(\begin{array}{c}
A_{j,1}\\
A_{j,2}\\
A_{j,2}
\end{array}\right),\qquad\boldsymbol{\gamma}_{j}=\left(\begin{array}{c}
\gamma_{j,1}\\
\gamma_{j,2}\\
\gamma_{j,3}
\end{array}\right),
\end{equation}
we can write this as
\begin{equation}
V_{M,j}=-i\frac{1}{2}E_{M}\boldsymbol{\gamma}_{j}\cdot(\mathbf{A}_{j}\times\boldsymbol{\gamma}_{j}).
\end{equation}
Consider the decomposition using three mutually orthonormal unit vectors $\hat{\mathbf{A}}_{j}=\mathbf{A}_{j}/|\mathbf{A}_{j}| $ and $\hat{\mathbf{A}}_{j,\pm}$:
\begin{equation}
\boldsymbol{\gamma}_{j}=\hat{\mathbf{A}}_{j}\gamma_{j,0}+\hat{\mathbf{A}}_{j,+}\gamma_{j,+}+\hat{\mathbf{A}}_{j,-}\gamma_{j,-}.
\end{equation}
The orthonormality of the vectors implies that $\gamma_{j,0}$ and $\gamma_{j,\pm}$ are Majorana operators.
We find that $\gamma_{j,0}$ is a MZM; it cancels from the Hamiltonian because
\begin{equation}
\mathbf{A}_{j}\times\hat{\mathbf{A}}_{j}=0,\qquad\hat{\mathbf{A}}_{j}\cdot(\mathbf{A}_{j}\times\hat{\mathbf{A}}_{j,\pm})=0.
\end{equation}
We have
\begin{align}
&V_{M,j}=
iE_{M}\gamma_{j,+}\gamma_{j,-}\,\mathbf{A}_{j}\cdot(\hat{\mathbf{A}}_{j,+}\times\hat{\mathbf{A}}_{j,-}).
\end{align}
We use fermion conventions where $i\gamma_{j,+}\gamma_{j,-}=-1$ in the groundstate; this corresponds to the orientation 
\begin{equation}
\mathbf{A}_{j}\cdot(\hat{\mathbf{A}}_{j,+}\times\hat{\mathbf{A}}_{j,-})>0\quad\Rightarrow\quad\hat{\mathbf{A}}_{j,+}\times\hat{\mathbf{A}}_{j,-}=\hat{\mathbf{A}}_{j}.
\end{equation}

Working in the regime $k_BT,\delta\varepsilon_{0,1} \ll E_M$, we can project to the $i\gamma_{j,+}\gamma_{j,-}=-1$ sector.
Denoting this low-energy projection by $\mapsto$, we find
\begin{equation}
\boldsymbol{\gamma}_{j}\mapsto\hat{\mathbf{A}}_{j}\gamma_{j,0},
\end{equation}
\begin{equation}
i\gamma_{a}\gamma_{b} \mapsto -\sum_c \epsilon_{abc}\, \mathbf{e}_c\cdot \hat{\mathbf{A}}_j,
\end{equation}
where $\mathbf{e}_1 = (1,0,0)^\top$, $\mathbf{e}_2 = (0,1,0)^\top$ and $\mathbf{e}_3 = (0,0,1)^\top$.

Using these projections, we work to first-order in $\delta_\pm /E_M$ and obtain a low-energy effective form of $H_{tt}$:
\begin{equation}
    H_\text{eff} = \sigma_z \left[ \tfrac{1}{2}\hbar\Omega_0 + Q\delta_+ \right] + Q\delta_-
\end{equation}
where $Q = P_- \prod_{k\in S_\mathcal{B}} \mathcal{P}_k P_-^\dagger$ and $P_- = \prod_j \tfrac{1}{2}(1-i\gamma_{j,+} \gamma_{j,-}).$

The total system, including the resonator, the top-transmon, and the coupling between the two, is described by the Jaynes-Cummings Hamiltonian
\begin{equation}\label{eqn:H_JC}
    H_{JC} = H_\text{eff} + \hbar \omega_0 (a^\dagger a + \tfrac{1}{2}) + \hbar g \left(a\sigma_+ + a^\dagger \sigma_-\right),
\end{equation}
where $a$ and $a^\dagger$ are photon ladder operators, $\omega_0$ is the resonator's bare resonance frequency, $g$ is the transmon-resonator coupling strength, and $\sigma_\pm = (\sigma_x\pm i \sigma_y)/2$.
The resonator operates in the dispersive regime~\cite{Transmon_Qubit2007,Dispersive_Regime_2007}, $\delta\omega^2  \gg g^2 (n+1)$, where $\delta\omega = \Omega_0 - \omega_0$ and $n$ is the resonator's photon occupation number.
In this limit, one can diagonalize Hamiltonian~\ref{eqn:H_JC} within a block spanned by states $|n,m=1\rangle$ and $|n+1,m=0\rangle$, where $m$ is the transmon level.
We thus obtain effective resonance frequencies to second order in $g/\delta\omega$ (neglecting virtual transitions to $m>1$ levels~\cite{Transmon_Qubit2007}), assuming the transmon is in the lowest energy level:
\begin{equation}
    \omega_\text{eff} = \omega_0 - \frac{\hbar g^2}{\hbar\delta\omega + 2\delta_+ Q}
\end{equation}
A change in the eigenvalue of $Q$ results in a shift in this resonance frequency approximately given by Eq.(3) of the main text, where the constant $C = 4g^2/(\hbar \delta\omega^2)$.


\begin{thebibliography}{95}%
\makeatletter
\providecommand \@ifxundefined [1]{%
 \@ifx{#1\undefined}
}%
\providecommand \@ifnum [1]{%
 \ifnum #1\expandafter \@firstoftwo
 \else \expandafter \@secondoftwo
 \fi
}%
\providecommand \@ifx [1]{%
 \ifx #1\expandafter \@firstoftwo
 \else \expandafter \@secondoftwo
 \fi
}%
\providecommand \natexlab [1]{#1}%
\providecommand \enquote  [1]{``#1''}%
\providecommand \bibnamefont  [1]{#1}%
\providecommand \bibfnamefont [1]{#1}%
\providecommand \citenamefont [1]{#1}%
\providecommand \href@noop [0]{\@secondoftwo}%
\providecommand \href [0]{\begingroup \@sanitize@url \@href}%
\providecommand \@href[1]{\@@startlink{#1}\@@href}%
\providecommand \@@href[1]{\endgroup#1\@@endlink}%
\providecommand \@sanitize@url [0]{\catcode `\\12\catcode `\$12\catcode
  `\&12\catcode `\#12\catcode `\^12\catcode `\_12\catcode `\%12\relax}%
\providecommand \@@startlink[1]{}%
\providecommand \@@endlink[0]{}%
\providecommand \url  [0]{\begingroup\@sanitize@url \@url }%
\providecommand \@url [1]{\endgroup\@href {#1}{\urlprefix }}%
\providecommand \urlprefix  [0]{URL }%
\providecommand \Eprint [0]{\href }%
\providecommand \doibase [0]{http://dx.doi.org/}%
\providecommand \selectlanguage [0]{\@gobble}%
\providecommand \bibinfo  [0]{\@secondoftwo}%
\providecommand \bibfield  [0]{\@secondoftwo}%
\providecommand \translation [1]{[#1]}%
\providecommand \BibitemOpen [0]{}%
\providecommand \bibitemStop [0]{}%
\providecommand \bibitemNoStop [0]{.\EOS\space}%
\providecommand \EOS [0]{\spacefactor3000\relax}%
\providecommand \BibitemShut  [1]{\csname bibitem#1\endcsname}%
\let\auto@bib@innerbib\@empty
%</preamble>
\bibitem [{\citenamefont {Bravyi}\ \emph {et~al.}(2016)\citenamefont {Bravyi},
  \citenamefont {Smith},\ and\ \citenamefont {Smolin}}]{Bravyi_PBC2016}%
  \BibitemOpen
  \bibfield  {author} {\bibinfo {author} {\bibfnamefont {S.}~\bibnamefont
  {Bravyi}}, \bibinfo {author} {\bibfnamefont {G.}~\bibnamefont {Smith}}, \
  and\ \bibinfo {author} {\bibfnamefont {J.~A.}\ \bibnamefont {Smolin}},\
  }\href {http://dx.doi.org/10.1103/PhysRevX.6.021043} {\bibfield  {journal}
  {\bibinfo  {journal} {Phys. Rev. X}\ }\textbf {\bibinfo {volume} {6}},\
  \bibinfo {pages} {021043} (\bibinfo {year} {2016})}\BibitemShut {NoStop}%
\bibitem [{\citenamefont {Gottesman}(1997)}]{Gottesman_thesis}%
  \BibitemOpen
  \bibfield  {author} {\bibinfo {author} {\bibfnamefont {D.}~\bibnamefont
  {Gottesman}},\ }\emph {\bibinfo {title} {Stabilizer Codes and Quantum Error
  Correction}},\ \href@noop {} {Ph.D. thesis},\ \bibinfo  {school} {CalTech}
  (\bibinfo {year} {1997})\BibitemShut {NoStop}%
\bibitem [{\citenamefont {Bravyi}\ and\ \citenamefont
  {Kitaev}(2002)}]{Bravyi_Kitaev_Fermionic_QC2002}%
  \BibitemOpen
  \bibfield  {author} {\bibinfo {author} {\bibfnamefont {S.~B.}\ \bibnamefont
  {Bravyi}}\ and\ \bibinfo {author} {\bibfnamefont {A.~Y.}\ \bibnamefont
  {Kitaev}},\ }\href {http://dx.doi.org/10.1006/aphy.2002.6254} {\bibfield
  {journal} {\bibinfo  {journal} {Ann. Phys. (N. Y.)}\ }\textbf {\bibinfo
  {volume} {298}},\ \bibinfo {pages} {210–226} (\bibinfo {year}
  {2002})}\BibitemShut {NoStop}%
\bibitem [{\citenamefont {O’Brien}\ \emph {et~al.}(2018)\citenamefont
  {O’Brien}, \citenamefont {Rożek},\ and\ \citenamefont
  {Akhmerov}}]{Maj_Based_FQC2018}%
  \BibitemOpen
  \bibfield  {author} {\bibinfo {author} {\bibfnamefont {T.~E.}~\bibnamefont
  {O’Brien}}, \bibinfo {author} {\bibfnamefont {P.}~\bibnamefont {Rożek}}, \
  and\ \bibinfo {author} {\bibfnamefont {A.~R.}~\bibnamefont {Akhmerov}},\ }\href
  {http://dx.doi.org/10.1103/PhysRevLett.120.220504} {\bibfield  {journal}
  {\bibinfo  {journal} {Phys. Rev. Lett.}\ }\textbf {\bibinfo {volume} {120}},\
  \bibinfo {pages} {220504} (\bibinfo {year} {2018})}\BibitemShut {NoStop}%
\bibitem [{\citenamefont {{Kitaev}}(2001)}]{Kitaev_chain2001}%
  \BibitemOpen
  \bibfield  {author} {\bibinfo {author} {\bibfnamefont {A.~Y.}\ \bibnamefont
  {{Kitaev}}},\ }\href {\doibase 10.1070/1063-7869/44/10S/S29} {\bibfield
  {journal} {\bibinfo  {journal} {Phys.-Uspekhi}\ }\textbf {\bibinfo {volume}
  {44}},\ \bibinfo {pages} {131} (\bibinfo {year} {2001})}\BibitemShut
  {NoStop}%
\bibitem [{\citenamefont {Kitaev}(2006)}]{Kitaev_Anyons2006}%
  \BibitemOpen
  \bibfield  {author} {\bibinfo {author} {\bibfnamefont {A.~Y.}~\bibnamefont
  {Kitaev}},\ }\href {http://dx.doi.org/10.1016/j.aop.2005.10.005} {\bibfield
  {journal} {\bibinfo  {journal} {Ann. Phys. (N. Y.)}\ }\textbf {\bibinfo
  {volume} {321}},\ \bibinfo {pages} {2–111} (\bibinfo {year}
  {2006})}\BibitemShut {NoStop}%
\bibitem [{\citenamefont {Sarma}\ \emph {et~al.}(2015)\citenamefont {Sarma},
  \citenamefont {Freedman},\ and\ \citenamefont {Nayak}}]{MZM_TQC_Review2015}%
  \BibitemOpen
  \bibfield  {author} {\bibinfo {author} {\bibfnamefont {S.~D.}\ \bibnamefont
  {Sarma}}, \bibinfo {author} {\bibfnamefont {M.}~\bibnamefont {Freedman}}, \
  and\ \bibinfo {author} {\bibfnamefont {C.}~\bibnamefont {Nayak}},\ }\href
  {https://doi.org/10.1038/npjqi.2015.1} {\bibfield  {journal} {\bibinfo
  {journal} {npj Quantum Inf.}\ }\textbf {\bibinfo {volume} {1}},\ \bibinfo
  {pages} {15001} (\bibinfo {year} {2015})}\BibitemShut {NoStop}%
\bibitem [{\citenamefont {Aasen}\ \emph {et~al.}(2016)\citenamefont {Aasen},
  \citenamefont {Hell}, \citenamefont {Mishmash}, \citenamefont {Higginbotham},
  \citenamefont {Danon}, \citenamefont {Leijnse}, \citenamefont {Jespersen},
  \citenamefont {Folk}, \citenamefont {Marcus}, \citenamefont {Flensberg},\
  and\ \citenamefont {Alicea}}]{Milestones_Paper2016}%
  \BibitemOpen
  \bibfield  {author} {\bibinfo {author} {\bibfnamefont {D.}~\bibnamefont
  {Aasen}}, \bibinfo {author} {\bibfnamefont {M.}~\bibnamefont {Hell}},
  \bibinfo {author} {\bibfnamefont {R.~V.}\ \bibnamefont {Mishmash}}, \bibinfo
  {author} {\bibfnamefont {A.}~\bibnamefont {Higginbotham}}, \bibinfo {author}
  {\bibfnamefont {J.}~\bibnamefont {Danon}}, \bibinfo {author} {\bibfnamefont
  {M.}~\bibnamefont {Leijnse}}, \bibinfo {author} {\bibfnamefont {T.~S.}\
  \bibnamefont {Jespersen}}, \bibinfo {author} {\bibfnamefont {J.~A.}\
  \bibnamefont {Folk}}, \bibinfo {author} {\bibfnamefont {C.~M.}\ \bibnamefont
  {Marcus}}, \bibinfo {author} {\bibfnamefont {K.}~\bibnamefont {Flensberg}}, \
  and\ \bibinfo {author} {\bibfnamefont {J.}~\bibnamefont {Alicea}},\ }\href
  {\doibase 10.1103/PhysRevX.6.031016} {\bibfield  {journal} {\bibinfo
  {journal} {Phys. Rev. X}\ }\textbf {\bibinfo {volume} {6}},\ \bibinfo {pages}
  {031016} (\bibinfo {year} {2016})}\BibitemShut {NoStop}%
\bibitem [{\citenamefont {Karzig}\ \emph {et~al.}(2017)\citenamefont {Karzig},
  \citenamefont {Knapp}, \citenamefont {Lutchyn}, \citenamefont {Bonderson},
  \citenamefont {Hastings}, \citenamefont {Nayak}, \citenamefont {Alicea},
  \citenamefont {Flensberg}, \citenamefont {Plugge}, \citenamefont {Oreg},
  \citenamefont {Marcus},\ and\ \citenamefont
  {Freedman}}]{Karzig_Majorana2017}%
  \BibitemOpen
  \bibfield  {author} {\bibinfo {author} {\bibfnamefont {T.}~\bibnamefont
  {Karzig}}, \bibinfo {author} {\bibfnamefont {C.}~\bibnamefont {Knapp}},
  \bibinfo {author} {\bibfnamefont {R.~M.}\ \bibnamefont {Lutchyn}}, \bibinfo
  {author} {\bibfnamefont {P.}~\bibnamefont {Bonderson}}, \bibinfo {author}
  {\bibfnamefont {M.~B.}\ \bibnamefont {Hastings}}, \bibinfo {author}
  {\bibfnamefont {C.}~\bibnamefont {Nayak}}, \bibinfo {author} {\bibfnamefont
  {J.}~\bibnamefont {Alicea}}, \bibinfo {author} {\bibfnamefont
  {K.}~\bibnamefont {Flensberg}}, \bibinfo {author} {\bibfnamefont
  {S.}~\bibnamefont {Plugge}}, \bibinfo {author} {\bibfnamefont
  {Y.}~\bibnamefont {Oreg}}, \bibinfo {author} {\bibfnamefont {C.~M.}\
  \bibnamefont {Marcus}}, \ and\ \bibinfo {author} {\bibfnamefont {M.~H.}\
  \bibnamefont {Freedman}},\ }\href
  {http://dx.doi.org/10.1103/PhysRevB.95.235305} {\bibfield  {journal}
  {\bibinfo  {journal} {Phys. Rev. B}\ }\textbf {\bibinfo {volume} {95}},\
  \bibinfo {pages} {235305} (\bibinfo {year} {2017})}\BibitemShut {NoStop}%
\bibitem [{\citenamefont {Bonderson}\ \emph {et~al.}(2008)\citenamefont
  {Bonderson}, \citenamefont {Freedman},\ and\ \citenamefont
  {Nayak}}]{Mmt_Only_QC_2008}%
  \BibitemOpen
  \bibfield  {author} {\bibinfo {author} {\bibfnamefont {P.}~\bibnamefont
  {Bonderson}}, \bibinfo {author} {\bibfnamefont {M.}~\bibnamefont {Freedman}},
  \ and\ \bibinfo {author} {\bibfnamefont {C.}~\bibnamefont {Nayak}},\ }\href
  {http://dx.doi.org/10.1103/PhysRevLett.101.010501} {\bibfield  {journal}
  {\bibinfo  {journal} {Phys. Rev. Lett.}\ }\textbf {\bibinfo {volume} {101}},\
  \bibinfo {pages} {010501} (\bibinfo {year} {2008})}\BibitemShut {NoStop}%
\bibitem [{\citenamefont {Zheng}\ \emph {et~al.}(2017)\citenamefont {Zheng},
  \citenamefont {Dua},\ and\ \citenamefont
  {Jiang}}]{Mmt_only_no_forced_mmts_2017}%
  \BibitemOpen
  \bibfield  {author} {\bibinfo {author} {\bibfnamefont {H.}~\bibnamefont
  {Zheng}}, \bibinfo {author} {\bibfnamefont {A.}~\bibnamefont {Dua}}, \ and\
  \bibinfo {author} {\bibfnamefont {L.}~\bibnamefont {Jiang}},\ }\href
  {http://dx.doi.org/10.1088/1367-2630/aa50bb} {\bibfield  {journal} {\bibinfo
  {journal} {New J. Phys.}\ }\textbf {\bibinfo {volume} {18}},\ \bibinfo
  {pages} {123027} (\bibinfo {year} {2017})}\BibitemShut {NoStop}%
\bibitem [{\citenamefont {Bomantara}\ and\ \citenamefont
  {Gong}(2020)}]{Mmt_only_QC_corner_modes_2020}%
  \BibitemOpen
  \bibfield  {author} {\bibinfo {author} {\bibfnamefont {R.~W.}\ \bibnamefont
  {Bomantara}}\ and\ \bibinfo {author} {\bibfnamefont {J.}~\bibnamefont
  {Gong}},\ }\href {http://dx.doi.org/10.1103/PhysRevB.101.085401} {\bibfield
  {journal} {\bibinfo  {journal} {Phys. Rev. B}\ }\textbf {\bibinfo {volume}
  {101}},\ \bibinfo {pages} {085401} (\bibinfo {year} {2020})}\BibitemShut
  {NoStop}%
\bibitem [{\citenamefont {Tran}\ \emph {et~al.}(2020)\citenamefont {Tran},
  \citenamefont {Bocharov}, \citenamefont {Bauer},\ and\ \citenamefont
  {Bonderson}}]{Optimizing_Clifford_Gates_2020}%
  \BibitemOpen
  \bibfield  {author} {\bibinfo {author} {\bibfnamefont {A.}~\bibnamefont
  {Tran}}, \bibinfo {author} {\bibfnamefont {A.}~\bibnamefont {Bocharov}},
  \bibinfo {author} {\bibfnamefont {B.}~\bibnamefont {Bauer}}, \ and\ \bibinfo
  {author} {\bibfnamefont {P.}~\bibnamefont {Bonderson}},\ }\href
  {http://dx.doi.org/10.21468/SciPostPhys.8.6.091} {\bibfield  {journal}
  {\bibinfo  {journal} {SciPost Phys.}\ }\textbf {\bibinfo {volume} {8}},\
  \bibinfo {pages} {91} (\bibinfo {year} {2020})}\BibitemShut {NoStop}%
\bibitem [{\citenamefont {Alicea}\ \emph {et~al.}(2011)\citenamefont {Alicea},
  \citenamefont {Oreg}, \citenamefont {Refael}, \citenamefont {von Oppen},\
  and\ \citenamefont {Fisher}}]{Majoranas_Tjunctions2011}%
  \BibitemOpen
  \bibfield  {author} {\bibinfo {author} {\bibfnamefont {J.}~\bibnamefont
  {Alicea}}, \bibinfo {author} {\bibfnamefont {Y.}~\bibnamefont {Oreg}},
  \bibinfo {author} {\bibfnamefont {G.}~\bibnamefont {Refael}}, \bibinfo
  {author} {\bibfnamefont {F.}~\bibnamefont {von Oppen}}, \ and\ \bibinfo
  {author} {\bibfnamefont {M.~P.~A.}\ \bibnamefont {Fisher}},\ }\href
  {http://dx.doi.org/10.1038/nphys1915} {\bibfield  {journal} {\bibinfo
  {journal} {Nat. Phys.}\ }\textbf {\bibinfo {volume} {7}},\ \bibinfo {pages}
  {412–417} (\bibinfo {year} {2011})}\BibitemShut {NoStop}%
\bibitem [{\citenamefont {van Heck}\ \emph {et~al.}(2012)\citenamefont {van
  Heck}, \citenamefont {Akhmerov}, \citenamefont {Hassler}, \citenamefont
  {Burrello},\ and\ \citenamefont {Beenakker}}]{Coulomb_assisted_braiding2012}%
  \BibitemOpen
  \bibfield  {author} {\bibinfo {author} {\bibfnamefont {B.}~\bibnamefont {van
  Heck}}, \bibinfo {author} {\bibfnamefont {A.~R.}\ \bibnamefont {Akhmerov}},
  \bibinfo {author} {\bibfnamefont {F.}~\bibnamefont {Hassler}}, \bibinfo
  {author} {\bibfnamefont {M.}~\bibnamefont {Burrello}}, \ and\ \bibinfo
  {author} {\bibfnamefont {C.~W.~J.}\ \bibnamefont {Beenakker}},\ }\href
  {http://dx.doi.org/10.1088/1367-2630/14/3/035019} {\bibfield  {journal}
  {\bibinfo  {journal} {New J. Phys.}\ }\textbf {\bibinfo {volume} {14}},\
  \bibinfo {pages} {035019} (\bibinfo {year} {2012})}\BibitemShut {NoStop}%
\bibitem [{\citenamefont {Hell}\ \emph {et~al.}(2016)\citenamefont {Hell},
  \citenamefont {Danon}, \citenamefont {Flensberg},\ and\ \citenamefont
  {Leijnse}}]{Time_Scales_Coulomb_Block2016}%
  \BibitemOpen
  \bibfield  {author} {\bibinfo {author} {\bibfnamefont {M.}~\bibnamefont
  {Hell}}, \bibinfo {author} {\bibfnamefont {J.}~\bibnamefont {Danon}},
  \bibinfo {author} {\bibfnamefont {K.}~\bibnamefont {Flensberg}}, \ and\
  \bibinfo {author} {\bibfnamefont {M.}~\bibnamefont {Leijnse}},\ }\href
  {\doibase 10.1103/PhysRevB.94.035424} {\bibfield  {journal} {\bibinfo
  {journal} {Phys. Rev. B}\ }\textbf {\bibinfo {volume} {94}},\ \bibinfo
  {pages} {035424} (\bibinfo {year} {2016})}\BibitemShut {NoStop}%
\bibitem [{\citenamefont {Knapp}\ \emph {et~al.}(2016)\citenamefont {Knapp},
  \citenamefont {Zaletel}, \citenamefont {Liu}, \citenamefont {Cheng},
  \citenamefont {Bonderson},\ and\ \citenamefont
  {Nayak}}]{Diabatic_errors2016}%
  \BibitemOpen
  \bibfield  {author} {\bibinfo {author} {\bibfnamefont {C.}~\bibnamefont
  {Knapp}}, \bibinfo {author} {\bibfnamefont {M.}~\bibnamefont {Zaletel}},
  \bibinfo {author} {\bibfnamefont {D.~E.}\ \bibnamefont {Liu}}, \bibinfo
  {author} {\bibfnamefont {M.}~\bibnamefont {Cheng}}, \bibinfo {author}
  {\bibfnamefont {P.}~\bibnamefont {Bonderson}}, \ and\ \bibinfo {author}
  {\bibfnamefont {C.}~\bibnamefont {Nayak}},\ }\href {\doibase
  10.1103/PhysRevX.6.041003} {\bibfield  {journal} {\bibinfo  {journal} {Phys.
  Rev. X}\ }\textbf {\bibinfo {volume} {6}},\ \bibinfo {pages} {041003}
  (\bibinfo {year} {2016})}\BibitemShut {NoStop}%
\bibitem [{\citenamefont {Koch}\ \emph {et~al.}(2007)\citenamefont {Koch},
  \citenamefont {Yu}, \citenamefont {Gambetta}, \citenamefont {Houck},
  \citenamefont {Schuster}, \citenamefont {Majer}, \citenamefont {Blais},
  \citenamefont {Devoret}, \citenamefont {Girvin},\ and\ \citenamefont
  {Schoelkopf}}]{Transmon_Qubit2007}%
  \BibitemOpen
  \bibfield  {author} {\bibinfo {author} {\bibfnamefont {J.}~\bibnamefont
  {Koch}}, \bibinfo {author} {\bibfnamefont {T.~M.}\ \bibnamefont {Yu}},
  \bibinfo {author} {\bibfnamefont {J.}~\bibnamefont {Gambetta}}, \bibinfo
  {author} {\bibfnamefont {A.~A.}\ \bibnamefont {Houck}}, \bibinfo {author}
  {\bibfnamefont {D.~I.}\ \bibnamefont {Schuster}}, \bibinfo {author}
  {\bibfnamefont {J.}~\bibnamefont {Majer}}, \bibinfo {author} {\bibfnamefont
  {A.}~\bibnamefont {Blais}}, \bibinfo {author} {\bibfnamefont {M.~H.}\
  \bibnamefont {Devoret}}, \bibinfo {author} {\bibfnamefont {S.~M.}\
  \bibnamefont {Girvin}}, \ and\ \bibinfo {author} {\bibfnamefont {R.~J.}\
  \bibnamefont {Schoelkopf}},\ }\href
  {http://dx.doi.org/10.1103/PhysRevA.76.042319} {\bibfield  {journal}
  {\bibinfo  {journal} {Phys. Rev. A}\ }\textbf {\bibinfo {volume} {76}},\
  \bibinfo {pages} {042319} (\bibinfo {year} {2007})}\BibitemShut {NoStop}%
\bibitem [{\citenamefont {Schreier}\ \emph {et~al.}(2008)\citenamefont
  {Schreier}, \citenamefont {Houck}, \citenamefont {Koch}, \citenamefont
  {Schuster}, \citenamefont {Johnson}, \citenamefont {Chow}, \citenamefont
  {Gambetta}, \citenamefont {Majer}, \citenamefont {Frunzio}, \citenamefont
  {Devoret}, \citenamefont {Girvin},\ and\ \citenamefont
  {Schoelkopf}}]{Transmon_Experiment2008}%
  \BibitemOpen
  \bibfield  {author} {\bibinfo {author} {\bibfnamefont {J.~A.}\ \bibnamefont
  {Schreier}}, \bibinfo {author} {\bibfnamefont {A.~A.}\ \bibnamefont {Houck}},
  \bibinfo {author} {\bibfnamefont {J.}~\bibnamefont {Koch}}, \bibinfo {author}
  {\bibfnamefont {D.~I.}\ \bibnamefont {Schuster}}, \bibinfo {author}
  {\bibfnamefont {B.~R.}\ \bibnamefont {Johnson}}, \bibinfo {author}
  {\bibfnamefont {J.~M.}\ \bibnamefont {Chow}}, \bibinfo {author}
  {\bibfnamefont {J.~M.}\ \bibnamefont {Gambetta}}, \bibinfo {author}
  {\bibfnamefont {J.}~\bibnamefont {Majer}}, \bibinfo {author} {\bibfnamefont
  {L.}~\bibnamefont {Frunzio}}, \bibinfo {author} {\bibfnamefont {M.~H.}\
  \bibnamefont {Devoret}}, \bibinfo {author} {\bibfnamefont {S.~M.}\
  \bibnamefont {Girvin}}, \ and\ \bibinfo {author} {\bibfnamefont {R.~J.}\
  \bibnamefont {Schoelkopf}},\ }\href {\doibase 10.1103/PhysRevB.77.180502}
  {\bibfield  {journal} {\bibinfo  {journal} {Phys. Rev. B}\ }\textbf {\bibinfo
  {volume} {77}},\ \bibinfo {pages} {180502(R)} (\bibinfo {year}
  {2008})}\BibitemShut {NoStop}%
\bibitem [{\citenamefont {Hassler}\ \emph {et~al.}(2011)\citenamefont
  {Hassler}, \citenamefont {Akhmerov},\ and\ \citenamefont
  {Beenakker}}]{Top_Transmon2011}%
  \BibitemOpen
  \bibfield  {author} {\bibinfo {author} {\bibfnamefont {F.}~\bibnamefont
  {Hassler}}, \bibinfo {author} {\bibfnamefont {A.~R.}\ \bibnamefont
  {Akhmerov}}, \ and\ \bibinfo {author} {\bibfnamefont {C.~W.~J.}\ \bibnamefont
  {Beenakker}},\ }\href {\doibase 10.1088/1367-2630/13/9/095004} {\bibfield
  {journal} {\bibinfo  {journal} {New J. Phys.}\ }\textbf {\bibinfo {volume}
  {13}},\ \bibinfo {pages} {095004} (\bibinfo {year} {2011})}\BibitemShut
  {NoStop}%
\bibitem [{\citenamefont {Hyart}\ \emph {et~al.}(2013)\citenamefont {Hyart},
  \citenamefont {van Heck}, \citenamefont {Fulga}, \citenamefont {Burrello},
  \citenamefont {Akhmerov},\ and\ \citenamefont
  {Beenakker}}]{Majorana_Ram_Flux_control2013}%
  \BibitemOpen
  \bibfield  {author} {\bibinfo {author} {\bibfnamefont {T.}~\bibnamefont
  {Hyart}}, \bibinfo {author} {\bibfnamefont {B.}~\bibnamefont {van Heck}},
  \bibinfo {author} {\bibfnamefont {I.~C.}\ \bibnamefont {Fulga}}, \bibinfo
  {author} {\bibfnamefont {M.}~\bibnamefont {Burrello}}, \bibinfo {author}
  {\bibfnamefont {A.~R.}\ \bibnamefont {Akhmerov}}, \ and\ \bibinfo {author}
  {\bibfnamefont {C.~W.~J.}\ \bibnamefont {Beenakker}},\ }\href
  {http://dx.doi.org/10.1103/PhysRevB.88.035121} {\bibfield  {journal}
  {\bibinfo  {journal} {Phys. Rev. B}\ }\textbf {\bibinfo {volume} {88}},\
  \bibinfo {pages} {035121} (\bibinfo {year} {2013})}\BibitemShut {NoStop}%
\bibitem [{\citenamefont {Yoganathan}\ \emph {et~al.}(2019)\citenamefont
  {Yoganathan}, \citenamefont {Jozsa},\ and\ \citenamefont
  {Strelchuk}}]{Mithuna_Magic_State_PBC2019}%
  \BibitemOpen
  \bibfield  {author} {\bibinfo {author} {\bibfnamefont {M.}~\bibnamefont
  {Yoganathan}}, \bibinfo {author} {\bibfnamefont {R.}~\bibnamefont {Jozsa}}, \
  and\ \bibinfo {author} {\bibfnamefont {S.}~\bibnamefont {Strelchuk}},\ }\href
  {http://dx.doi.org/10.1098/rspa.2018.0427} {\bibfield  {journal} {\bibinfo
  {journal} {Proc. Math. Phys. Eng. Sci.}\ }\textbf {\bibinfo {volume} {475}},\
  \bibinfo {pages} {20180427} (\bibinfo {year} {2019})}\BibitemShut {NoStop}%
\bibitem [{\citenamefont {Read}\ and\ \citenamefont
  {Green}(2000)}]{ReadGreen2000}%
  \BibitemOpen
  \bibfield  {author} {\bibinfo {author} {\bibfnamefont {N.}~\bibnamefont
  {Read}}\ and\ \bibinfo {author} {\bibfnamefont {D.}~\bibnamefont {Green}},\
  }\href {\doibase 10.1103/PhysRevB.61.10267} {\bibfield  {journal} {\bibinfo
  {journal} {Phys. Rev. B}\ }\textbf {\bibinfo {volume} {61}},\ \bibinfo
  {pages} {10267} (\bibinfo {year} {2000})}\BibitemShut {NoStop}%
\bibitem [{\citenamefont {Ivanov}(2001)}]{MZMs_in_pwave_superconds2001}%
  \BibitemOpen
  \bibfield  {author} {\bibinfo {author} {\bibfnamefont {D.~A.}\ \bibnamefont
  {Ivanov}},\ }\href {\doibase 10.1103/PhysRevLett.86.268} {\bibfield
  {journal} {\bibinfo  {journal} {Phys. Rev. Lett.}\ }\textbf {\bibinfo
  {volume} {86}},\ \bibinfo {pages} {268} (\bibinfo {year} {2001})}\BibitemShut
  {NoStop}%
\bibitem [{\citenamefont {Nayak}\ \emph {et~al.}(2008)\citenamefont {Nayak},
  \citenamefont {Simon}, \citenamefont {Stern}, \citenamefont {Freedman},\ and\
  \citenamefont {Das~Sarma}}]{Non-Abelian_TQC_Review}%
  \BibitemOpen
  \bibfield  {author} {\bibinfo {author} {\bibfnamefont {C.}~\bibnamefont
  {Nayak}}, \bibinfo {author} {\bibfnamefont {S.~H.}\ \bibnamefont {Simon}},
  \bibinfo {author} {\bibfnamefont {A.}~\bibnamefont {Stern}}, \bibinfo
  {author} {\bibfnamefont {M.}~\bibnamefont {Freedman}}, \ and\ \bibinfo
  {author} {\bibfnamefont {S.}~\bibnamefont {Das~Sarma}},\ }\href
  {http://dx.doi.org/10.1103/RevModPhys.80.1083} {\bibfield  {journal}
  {\bibinfo  {journal} {Rev. Mod. Phys.}\ }\textbf {\bibinfo {volume} {80}},\
  \bibinfo {pages} {1083–1159} (\bibinfo {year} {2008})}\BibitemShut
  {NoStop}%
\bibitem [{Log()}]{LogicalMajorana}%
  \BibitemOpen
  \href@noop {} {}\bibinfo {note} {For other uses of logical Majoranas, see
  Refs.~\cite{Akhmerov2010,Goldstein2012,Maj_Triangle_Code2018,Behrends2020}.}\BibitemShut
  {Stop}%
\bibitem [{\citenamefont {Bravyi}(2006)}]{Bravyi5p2_2006}%
  \BibitemOpen
  \bibfield  {author} {\bibinfo {author} {\bibfnamefont {S.}~\bibnamefont
  {Bravyi}},\ }\href {\doibase 10.1103/PhysRevA.73.042313} {\bibfield
  {journal} {\bibinfo  {journal} {Phys. Rev. A}\ }\textbf {\bibinfo {volume}
  {73}},\ \bibinfo {pages} {042313} (\bibinfo {year} {2006})}\BibitemShut
  {NoStop}%
\bibitem [{App()}]{Appfn}%
  \BibitemOpen
  \href@noop {} {}\bibinfo {note} {See Supplemental Material for
  details on the logical braid group; measurements generated by FPBC; the
  resource cost of FPBC in other setups; and the top-transmon Hamiltonian,
  tri-junctions and shift frequencies. Supplemental Material includes
  Refs.~\onlinecite{Maj_Box_Qubits2017, Quantum_Dot_Readout2020}}\BibitemShut
  {NoStop}%
\bibitem [{\citenamefont {Bravyi}\ and\ \citenamefont
  {Kitaev}(2005)}]{Magic_State_Dist}%
  \BibitemOpen
  \bibfield  {author} {\bibinfo {author} {\bibfnamefont {S.}~\bibnamefont
  {Bravyi}}\ and\ \bibinfo {author} {\bibfnamefont {A.}~\bibnamefont
  {Kitaev}},\ }\href {http://dx.doi.org/10.1103/PhysRevA.71.022316} {\bibfield
  {journal} {\bibinfo  {journal} {Phys. Rev. A}\ }\textbf {\bibinfo {volume}
  {71}},\ \bibinfo {pages} {022316} (\bibinfo {year} {2005})}\BibitemShut
  {NoStop}%
\bibitem [{\citenamefont {Souza}\ \emph {et~al.}(2011)\citenamefont {Souza},
  \citenamefont {Zhang}, \citenamefont {Ryan},\ and\ \citenamefont
  {Laflamme}}]{MSD_Experiment2011}%
  \BibitemOpen
  \bibfield  {author} {\bibinfo {author} {\bibfnamefont {A.~M.}\ \bibnamefont
  {Souza}}, \bibinfo {author} {\bibfnamefont {J.}~\bibnamefont {Zhang}},
  \bibinfo {author} {\bibfnamefont {C.~A.}\ \bibnamefont {Ryan}}, \ and\
  \bibinfo {author} {\bibfnamefont {R.}~\bibnamefont {Laflamme}},\ }\href
  {\doibase 10.1038/ncomms1166} {\bibfield  {journal} {\bibinfo  {journal}
  {Nat. Commun.}\ }\textbf {\bibinfo {volume} {2}},\ \bibinfo {pages} {169}
  (\bibinfo {year} {2011})}\BibitemShut {NoStop}%
\bibitem [{\citenamefont {Litinski}\ and\ \citenamefont {von
  Oppen}(2018)}]{QC_with_MFCs}%
  \BibitemOpen
  \bibfield  {author} {\bibinfo {author} {\bibfnamefont {D.}~\bibnamefont
  {Litinski}}\ and\ \bibinfo {author} {\bibfnamefont {F.}~\bibnamefont {von
  Oppen}},\ }\href {https://link.aps.org/doi/10.1103/PhysRevB.97.205404}
  {\bibfield  {journal} {\bibinfo  {journal} {Phys. Rev. B}\ }\textbf {\bibinfo
  {volume} {97}},\ \bibinfo {pages} {205404} (\bibinfo {year}
  {2018})}\BibitemShut {NoStop}%
\bibitem [{\citenamefont {Li}(2018)}]{Maj_Triangle_Code2018}%
  \BibitemOpen
  \bibfield  {author} {\bibinfo {author} {\bibfnamefont {Y.}~\bibnamefont
  {Li}},\ }\href {\doibase 10.1103/PhysRevA.98.012336} {\bibfield  {journal}
  {\bibinfo  {journal} {Phys. Rev. A}\ }\textbf {\bibinfo {volume} {98}},\
  \bibinfo {pages} {012336} (\bibinfo {year} {2018})}\BibitemShut {NoStop}%
\bibitem [{\citenamefont {Gau}\ \emph {et~al.}(2020)\citenamefont {Gau},
  \citenamefont {Egger}, \citenamefont {Zazunov},\ and\ \citenamefont
  {Gefen}}]{Dark_Space_Stab2020}%
  \BibitemOpen
  \bibfield  {author} {\bibinfo {author} {\bibfnamefont {M.}~\bibnamefont
  {Gau}}, \bibinfo {author} {\bibfnamefont {R.}~\bibnamefont {Egger}}, \bibinfo
  {author} {\bibfnamefont {A.}~\bibnamefont {Zazunov}}, \ and\ \bibinfo
  {author} {\bibfnamefont {Y.}~\bibnamefont {Gefen}},\ }\href {\doibase
  10.1103/PhysRevB.102.134501} {\bibfield  {journal} {\bibinfo  {journal}
  {Phys. Rev. B}\ }\textbf {\bibinfo {volume} {102}},\ \bibinfo {pages}
  {134501} (\bibinfo {year} {2020})}\BibitemShut {NoStop}%
\bibitem [{\citenamefont {Bravyi}\ and\ \citenamefont
  {Haah}(2012)}]{MSD_Low_Overhead2012}%
  \BibitemOpen
  \bibfield  {author} {\bibinfo {author} {\bibfnamefont {S.}~\bibnamefont
  {Bravyi}}\ and\ \bibinfo {author} {\bibfnamefont {J.}~\bibnamefont {Haah}},\
  }\href {\doibase 10.1103/PhysRevA.86.052329} {\bibfield  {journal} {\bibinfo
  {journal} {Phys. Rev. A}\ }\textbf {\bibinfo {volume} {86}},\ \bibinfo
  {pages} {052329} (\bibinfo {year} {2012})}\BibitemShut {NoStop}%
\bibitem [{\citenamefont {Jones}(2013)}]{MSD_Low_Overhead2013}%
  \BibitemOpen
  \bibfield  {author} {\bibinfo {author} {\bibfnamefont {C.}~\bibnamefont
  {Jones}},\ }\href {\doibase 10.1103/PhysRevA.87.042305} {\bibfield  {journal}
  {\bibinfo  {journal} {Phys. Rev. A}\ }\textbf {\bibinfo {volume} {87}},\
  \bibinfo {pages} {042305} (\bibinfo {year} {2013})}\BibitemShut {NoStop}%
\bibitem [{\citenamefont {Haah}\ \emph {et~al.}(2017)\citenamefont {Haah},
  \citenamefont {Hastings}, \citenamefont {Poulin},\ and\ \citenamefont
  {Wecker}}]{MSD_Low_Overhead2017}%
  \BibitemOpen
  \bibfield  {author} {\bibinfo {author} {\bibfnamefont {J.}~\bibnamefont
  {Haah}}, \bibinfo {author} {\bibfnamefont {M.~B.}\ \bibnamefont {Hastings}},
  \bibinfo {author} {\bibfnamefont {D.}~\bibnamefont {Poulin}}, \ and\ \bibinfo
  {author} {\bibfnamefont {D.}~\bibnamefont {Wecker}},\ }\href {\doibase
  10.22331/q-2017-10-03-31} {\bibfield  {journal} {\bibinfo  {journal}
  {Quantum}\ }\textbf {\bibinfo {volume} {1}},\ \bibinfo {pages} {31} (\bibinfo
  {year} {2017})}\BibitemShut {NoStop}%
\bibitem [{\citenamefont {Litinski}(2019)}]{MSD_Not_Costly2019}%
  \BibitemOpen
  \bibfield  {author} {\bibinfo {author} {\bibfnamefont {D.}~\bibnamefont
  {Litinski}},\ }\href {http://dx.doi.org/10.22331/q-2019-12-02-205} {\bibfield
   {journal} {\bibinfo  {journal} {Quantum}\ }\textbf {\bibinfo {volume} {3}},\
  \bibinfo {pages} {205} (\bibinfo {year} {2019})}\BibitemShut {NoStop}%
\bibitem [{\citenamefont {Karzig}\ \emph {et~al.}(2016)\citenamefont {Karzig},
  \citenamefont {Oreg}, \citenamefont {Refael},\ and\ \citenamefont
  {Freedman}}]{Karzig_geom_magic_2016}%
  \BibitemOpen
  \bibfield  {author} {\bibinfo {author} {\bibfnamefont {T.}~\bibnamefont
  {Karzig}}, \bibinfo {author} {\bibfnamefont {Y.}~\bibnamefont {Oreg}},
  \bibinfo {author} {\bibfnamefont {G.}~\bibnamefont {Refael}}, \ and\ \bibinfo
  {author} {\bibfnamefont {M.~H.}\ \bibnamefont {Freedman}},\ }\href {\doibase
  10.1103/PhysRevX.6.031019} {\bibfield  {journal} {\bibinfo  {journal} {Phys.
  Rev. X}\ }\textbf {\bibinfo {volume} {6}},\ \bibinfo {pages} {031019}
  (\bibinfo {year} {2016})}\BibitemShut {NoStop}%
\bibitem [{\citenamefont {Karzig}\ \emph {et~al.}(2019)\citenamefont {Karzig},
  \citenamefont {Oreg}, \citenamefont {Refael},\ and\ \citenamefont
  {Freedman}}]{Karzig_geom_magic_meas2019}%
  \BibitemOpen
  \bibfield  {author} {\bibinfo {author} {\bibfnamefont {T.}~\bibnamefont
  {Karzig}}, \bibinfo {author} {\bibfnamefont {Y.}~\bibnamefont {Oreg}},
  \bibinfo {author} {\bibfnamefont {G.}~\bibnamefont {Refael}}, \ and\ \bibinfo
  {author} {\bibfnamefont {M.~H.}\ \bibnamefont {Freedman}},\ }\href {\doibase
  10.1103/PhysRevB.99.144521} {\bibfield  {journal} {\bibinfo  {journal} {Phys.
  Rev. B}\ }\textbf {\bibinfo {volume} {99}},\ \bibinfo {pages} {144521}
  (\bibinfo {year} {2019})}\BibitemShut {NoStop}%
\bibitem [{\citenamefont {Nielsen}\ and\ \citenamefont
  {Chuang}(2010)}]{nielsen_chuang}%
  \BibitemOpen
  \bibfield  {author} {\bibinfo {author} {\bibfnamefont {M.~A.}\ \bibnamefont
  {Nielsen}}\ and\ \bibinfo {author} {\bibfnamefont {I.~L.}\ \bibnamefont
  {Chuang}},\ }\href@noop {} {\emph {\bibinfo {title} {Quantum Computation and
  Quantum Information: 10th Anniversary Edition}}}\ (\bibinfo  {publisher}
  {Cambridge University Press},\ \bibinfo {address} {Cambridge},\ \bibinfo
  {year} {2010})\BibitemShut {NoStop}%
\bibitem [{\citenamefont {Smith}\ \emph {et~al.}(2020)\citenamefont {Smith},
  \citenamefont {Cassidy}, \citenamefont {Reilly}, \citenamefont {Bartlett},\
  and\ \citenamefont {Grimsmo}}]{Smith_Bartlett_Readout2020}%
  \BibitemOpen
  \bibfield  {author} {\bibinfo {author} {\bibfnamefont {T.~B.}\ \bibnamefont
  {Smith}}, \bibinfo {author} {\bibfnamefont {M.~C.}\ \bibnamefont {Cassidy}},
  \bibinfo {author} {\bibfnamefont {D.~J.}\ \bibnamefont {Reilly}}, \bibinfo
  {author} {\bibfnamefont {S.~D.}\ \bibnamefont {Bartlett}}, \ and\ \bibinfo
  {author} {\bibfnamefont {A.~L.}\ \bibnamefont {Grimsmo}},\ }\href {\doibase
  10.1103/PRXQuantum.1.020313} {\bibfield  {journal} {\bibinfo  {journal} {PRX
  Quantum}\ }\textbf {\bibinfo {volume} {1}},\ \bibinfo {pages} {020313}
  (\bibinfo {year} {2020})}\BibitemShut {NoStop}%
\bibitem [{\citenamefont {Oreg}\ \emph {et~al.}(2010)\citenamefont {Oreg},
  \citenamefont {Refael},\ and\ \citenamefont {von
  Oppen}}]{von_Oppen_Majorana_Nanowires2010}%
  \BibitemOpen
  \bibfield  {author} {\bibinfo {author} {\bibfnamefont {Y.}~\bibnamefont
  {Oreg}}, \bibinfo {author} {\bibfnamefont {G.}~\bibnamefont {Refael}}, \ and\
  \bibinfo {author} {\bibfnamefont {F.}~\bibnamefont {von Oppen}},\ }\href
  {\doibase 10.1103/PhysRevLett.105.177002} {\bibfield  {journal} {\bibinfo
  {journal} {Phys. Rev. Lett.}\ }\textbf {\bibinfo {volume} {105}},\ \bibinfo
  {pages} {177002} (\bibinfo {year} {2010})}\BibitemShut {NoStop}%
\bibitem [{\citenamefont {Lutchyn}\ \emph {et~al.}(2010)\citenamefont
  {Lutchyn}, \citenamefont {Sau},\ and\ \citenamefont
  {Das~Sarma}}]{Das_Sarma_Majorana_Nanowire_2010}%
  \BibitemOpen
  \bibfield  {author} {\bibinfo {author} {\bibfnamefont {R.~M.}\ \bibnamefont
  {Lutchyn}}, \bibinfo {author} {\bibfnamefont {J.~D.}\ \bibnamefont {Sau}}, \
  and\ \bibinfo {author} {\bibfnamefont {S.}~\bibnamefont {Das~Sarma}},\ }\href
  {http://dx.doi.org/10.1103/PhysRevLett.105.077001} {\bibfield  {journal}
  {\bibinfo  {journal} {Phys. Rev. Lett.}\ }\textbf {\bibinfo {volume} {105}},\
  \bibinfo {pages} {077001} (\bibinfo {year} {2010})}\BibitemShut {NoStop}%
\bibitem [{\citenamefont {Alicea}(2012)}]{Alicea_Review_2012}%
  \BibitemOpen
  \bibfield  {author} {\bibinfo {author} {\bibfnamefont {J.}~\bibnamefont
  {Alicea}},\ }\href {\doibase 10.1088/0034-4885/75/7/076501} {\bibfield
  {journal} {\bibinfo  {journal} {Rep. Prog. Phys.}\ }\textbf {\bibinfo
  {volume} {75}},\ \bibinfo {pages} {076501} (\bibinfo {year}
  {2012})}\BibitemShut {NoStop}%
\bibitem [{\citenamefont {Das}\ \emph {et~al.}(2012)\citenamefont {Das},
  \citenamefont {Ronen}, \citenamefont {Most}, \citenamefont {Oreg},
  \citenamefont {Heiblum},\ and\ \citenamefont
  {Shtrikman}}]{Nature_Maj_Experiment2012}%
  \BibitemOpen
  \bibfield  {author} {\bibinfo {author} {\bibfnamefont {A.}~\bibnamefont
  {Das}}, \bibinfo {author} {\bibfnamefont {Y.}~\bibnamefont {Ronen}}, \bibinfo
  {author} {\bibfnamefont {Y.}~\bibnamefont {Most}}, \bibinfo {author}
  {\bibfnamefont {Y.}~\bibnamefont {Oreg}}, \bibinfo {author} {\bibfnamefont
  {M.}~\bibnamefont {Heiblum}}, \ and\ \bibinfo {author} {\bibfnamefont
  {H.}~\bibnamefont {Shtrikman}},\ }\href {http://dx.doi.org/10.1038/nphys2479}
  {\bibfield  {journal} {\bibinfo  {journal} {Nat. Phys.}\ }\textbf {\bibinfo
  {volume} {8}},\ \bibinfo {pages} {887–895} (\bibinfo {year}
  {2012})}\BibitemShut {NoStop}%
\bibitem [{\citenamefont {Mourik}\ \emph {et~al.}(2012)\citenamefont {Mourik},
  \citenamefont {Zuo}, \citenamefont {Frolov}, \citenamefont {Plissard},
  \citenamefont {Bakkers},\ and\ \citenamefont
  {Kouwenhoven}}]{Science_Maj_Experiment_2012}%
  \BibitemOpen
  \bibfield  {author} {\bibinfo {author} {\bibfnamefont {V.}~\bibnamefont
  {Mourik}}, \bibinfo {author} {\bibfnamefont {K.}~\bibnamefont {Zuo}},
  \bibinfo {author} {\bibfnamefont {S.~M.}\ \bibnamefont {Frolov}}, \bibinfo
  {author} {\bibfnamefont {S.~R.}\ \bibnamefont {Plissard}}, \bibinfo {author}
  {\bibfnamefont {E.~P. A.~M.}\ \bibnamefont {Bakkers}}, \ and\ \bibinfo
  {author} {\bibfnamefont {L.~P.}\ \bibnamefont {Kouwenhoven}},\ }\href
  {\doibase 10.1126/science.1222360} {\bibfield  {journal} {\bibinfo  {journal}
  {Science}\ }\textbf {\bibinfo {volume} {336}},\ \bibinfo {pages} {1003}
  (\bibinfo {year} {2012})}\BibitemShut {NoStop}%
\bibitem [{\citenamefont {Rokhinson}\ \emph {et~al.}(2012)\citenamefont
  {Rokhinson}, \citenamefont {Liu},\ and\ \citenamefont
  {Furdyna}}]{Fractional_AC_Josephson_Experiment_2012}%
  \BibitemOpen
  \bibfield  {author} {\bibinfo {author} {\bibfnamefont {L.~P.}\ \bibnamefont
  {Rokhinson}}, \bibinfo {author} {\bibfnamefont {X.}~\bibnamefont {Liu}}, \
  and\ \bibinfo {author} {\bibfnamefont {J.~K.}\ \bibnamefont {Furdyna}},\
  }\href {\doibase 10.1038/nphys2429} {\bibfield  {journal} {\bibinfo
  {journal} {Nat. Phys.}\ }\textbf {\bibinfo {volume} {8}},\ \bibinfo {pages}
  {795} (\bibinfo {year} {2012})}\BibitemShut {NoStop}%
\bibitem [{\citenamefont {Deng}\ \emph {et~al.}(2012)\citenamefont {Deng},
  \citenamefont {Yu}, \citenamefont {Huang}, \citenamefont {Larsson},
  \citenamefont {Caroff},\ and\ \citenamefont
  {Xu}}]{Maj_Experiment_Nano_Lett_2012}%
  \BibitemOpen
  \bibfield  {author} {\bibinfo {author} {\bibfnamefont {M.~T.}\ \bibnamefont
  {Deng}}, \bibinfo {author} {\bibfnamefont {C.~L.}\ \bibnamefont {Yu}},
  \bibinfo {author} {\bibfnamefont {G.~Y.}\ \bibnamefont {Huang}}, \bibinfo
  {author} {\bibfnamefont {M.}~\bibnamefont {Larsson}}, \bibinfo {author}
  {\bibfnamefont {P.}~\bibnamefont {Caroff}}, \ and\ \bibinfo {author}
  {\bibfnamefont {H.~Q.}\ \bibnamefont {Xu}},\ }\href {\doibase
  10.1021/nl303758w} {\bibfield  {journal} {\bibinfo  {journal} {Nano Lett.}\
  }\textbf {\bibinfo {volume} {12}},\ \bibinfo {pages} {6414} (\bibinfo {year}
  {2012})}\BibitemShut {NoStop}%
\bibitem [{\citenamefont {Beenakker}(2013)}]{Beenakker_Review_2013}%
  \BibitemOpen
  \bibfield  {author} {\bibinfo {author} {\bibfnamefont {C.}~\bibnamefont
  {Beenakker}},\ }\href {\doibase 10.1146/annurev-conmatphys-030212-184337}
  {\bibfield  {journal} {\bibinfo  {journal} {Annu. Rev. Condens. Matter
  Phys.}\ }\textbf {\bibinfo {volume} {4}},\ \bibinfo {pages} {113–136}
  (\bibinfo {year} {2013})}\BibitemShut {NoStop}%
\bibitem [{\citenamefont {Krogstrup}\ \emph {et~al.}(2015)\citenamefont
  {Krogstrup}, \citenamefont {Ziino}, \citenamefont {Chang}, \citenamefont
  {Albrecht}, \citenamefont {Madsen}, \citenamefont {Johnson}, \citenamefont
  {Nyg{\aa}rd}, \citenamefont {Marcus},\ and\ \citenamefont
  {Jespersen}}]{Maj_Epitaxy_2015}%
  \BibitemOpen
  \bibfield  {author} {\bibinfo {author} {\bibfnamefont {P.}~\bibnamefont
  {Krogstrup}}, \bibinfo {author} {\bibfnamefont {N.~L.~B.}\ \bibnamefont
  {Ziino}}, \bibinfo {author} {\bibfnamefont {W.}~\bibnamefont {Chang}},
  \bibinfo {author} {\bibfnamefont {S.~M.}\ \bibnamefont {Albrecht}}, \bibinfo
  {author} {\bibfnamefont {M.~H.}\ \bibnamefont {Madsen}}, \bibinfo {author}
  {\bibfnamefont {E.}~\bibnamefont {Johnson}}, \bibinfo {author} {\bibfnamefont
  {J.}~\bibnamefont {Nyg{\aa}rd}}, \bibinfo {author} {\bibfnamefont {C.~M.}\
  \bibnamefont {Marcus}}, \ and\ \bibinfo {author} {\bibfnamefont {T.~S.}\
  \bibnamefont {Jespersen}},\ }\href {\doibase 10.1038/nmat4176} {\bibfield
  {journal} {\bibinfo  {journal} {Nat. Mater.}\ }\textbf {\bibinfo {volume}
  {14}},\ \bibinfo {pages} {400} (\bibinfo {year} {2015})}\BibitemShut
  {NoStop}%
\bibitem [{\citenamefont {Albrecht}\ \emph {et~al.}(2016)\citenamefont
  {Albrecht}, \citenamefont {Higginbotham}, \citenamefont {Madsen},
  \citenamefont {Kuemmeth}, \citenamefont {Jespersen}, \citenamefont {Nygård},
  \citenamefont {Krogstrup},\ and\ \citenamefont
  {Marcus}}]{Albrecht-et-al-exponential-2016}%
  \BibitemOpen
  \bibfield  {author} {\bibinfo {author} {\bibfnamefont {S.~M.}\ \bibnamefont
  {Albrecht}}, \bibinfo {author} {\bibfnamefont {A.~P.}\ \bibnamefont
  {Higginbotham}}, \bibinfo {author} {\bibfnamefont {M.}~\bibnamefont
  {Madsen}}, \bibinfo {author} {\bibfnamefont {F.}~\bibnamefont {Kuemmeth}},
  \bibinfo {author} {\bibfnamefont {T.~S.}\ \bibnamefont {Jespersen}}, \bibinfo
  {author} {\bibfnamefont {J.}~\bibnamefont {Nygård}}, \bibinfo {author}
  {\bibfnamefont {P.}~\bibnamefont {Krogstrup}}, \ and\ \bibinfo {author}
  {\bibfnamefont {C.~M.}\ \bibnamefont {Marcus}},\ }\href
  {http://dx.doi.org/10.1038/nature17162} {\bibfield  {journal} {\bibinfo
  {journal} {Nature}\ }\textbf {\bibinfo {volume} {531}},\ \bibinfo {pages}
  {206} (\bibinfo {year} {2016})}\BibitemShut {NoStop}%
\bibitem [{\citenamefont {Chen}\ \emph {et~al.}(2017)\citenamefont {Chen},
  \citenamefont {Yu}, \citenamefont {Stenger}, \citenamefont {Hocevar},
  \citenamefont {Car}, \citenamefont {Plissard}, \citenamefont {Bakkers},
  \citenamefont {Stanescu},\ and\ \citenamefont
  {Frolov}}]{Maj_Experiment_2_2017}%
  \BibitemOpen
  \bibfield  {author} {\bibinfo {author} {\bibfnamefont {J.}~\bibnamefont
  {Chen}}, \bibinfo {author} {\bibfnamefont {P.}~\bibnamefont {Yu}}, \bibinfo
  {author} {\bibfnamefont {J.}~\bibnamefont {Stenger}}, \bibinfo {author}
  {\bibfnamefont {M.}~\bibnamefont {Hocevar}}, \bibinfo {author} {\bibfnamefont
  {D.}~\bibnamefont {Car}}, \bibinfo {author} {\bibfnamefont {S.~R.}\
  \bibnamefont {Plissard}}, \bibinfo {author} {\bibfnamefont {E.~P. A.~M.}\
  \bibnamefont {Bakkers}}, \bibinfo {author} {\bibfnamefont {T.~D.}\
  \bibnamefont {Stanescu}}, \ and\ \bibinfo {author} {\bibfnamefont {S.~M.}\
  \bibnamefont {Frolov}},\ }\href
  {https://advances.sciencemag.org/content/3/9/e1701476} {\bibfield  {journal}
  {\bibinfo  {journal} {Sci. Adv.}\ }\textbf {\bibinfo {volume} {3}},\ \bibinfo
  {pages} {e1701476} (\bibinfo {year} {2017})}\BibitemShut {NoStop}%
\bibitem [{\citenamefont {Lutchyn}\ \emph {et~al.}(2018)\citenamefont
  {Lutchyn}, \citenamefont {Bakkers}, \citenamefont {Kouwenhoven},
  \citenamefont {Krogstrup}, \citenamefont {Marcus},\ and\ \citenamefont
  {Oreg}}]{Lutchyn_Review_2018}%
  \BibitemOpen
  \bibfield  {author} {\bibinfo {author} {\bibfnamefont {R.~M.}\ \bibnamefont
  {Lutchyn}}, \bibinfo {author} {\bibfnamefont {E.~P. A.~M.}\ \bibnamefont
  {Bakkers}}, \bibinfo {author} {\bibfnamefont {L.~P.}\ \bibnamefont
  {Kouwenhoven}}, \bibinfo {author} {\bibfnamefont {P.}~\bibnamefont
  {Krogstrup}}, \bibinfo {author} {\bibfnamefont {C.~M.}\ \bibnamefont
  {Marcus}}, \ and\ \bibinfo {author} {\bibfnamefont {Y.}~\bibnamefont
  {Oreg}},\ }\href {\doibase 10.1038/s41578-018-0003-1} {\bibfield  {journal}
  {\bibinfo  {journal} {Nat. Rev. Mater.}\ }\textbf {\bibinfo {volume} {3}},\
  \bibinfo {pages} {52} (\bibinfo {year} {2018})}\BibitemShut {NoStop}%
\bibitem [{\citenamefont {Vaitiek{\.e}nas}\ \emph {et~al.}(2021)\citenamefont
  {Vaitiek{\.e}nas}, \citenamefont {Liu}, \citenamefont {Krogstrup},\ and\
  \citenamefont {Marcus}}]{vaitiekenas2021zero}%
  \BibitemOpen
  \bibfield  {author} {\bibinfo {author} {\bibfnamefont {S.}~\bibnamefont
  {Vaitiek{\.e}nas}}, \bibinfo {author} {\bibfnamefont {Y.}~\bibnamefont
  {Liu}}, \bibinfo {author} {\bibfnamefont {P.}~\bibnamefont {Krogstrup}}, \
  and\ \bibinfo {author} {\bibfnamefont {C.}~\bibnamefont {Marcus}},\ }\href
  {https://doi.org/10.1038/s41567-020-1017-3} {\bibfield  {journal} {\bibinfo
  {journal} {Nat. Phys.}\ }\textbf {\bibinfo {volume} {17}},\ \bibinfo {pages}
  {43} (\bibinfo {year} {2021})}\BibitemShut {NoStop}%
\bibitem [{\citenamefont {Flensberg}\ \emph {et~al.}(2021)\citenamefont
  {Flensberg}, \citenamefont {von Oppen},\ and\ \citenamefont
  {Stern}}]{Nanowire_Review_2021}%
  \BibitemOpen
  \bibfield  {author} {\bibinfo {author} {\bibfnamefont {K.}~\bibnamefont
  {Flensberg}}, \bibinfo {author} {\bibfnamefont {F.}~\bibnamefont {von
  Oppen}}, \ and\ \bibinfo {author} {\bibfnamefont {A.}~\bibnamefont {Stern}},\
  }\href {\doibase 10.1038/s41578-021-00336-6} {\bibfield  {journal} {\bibinfo
  {journal} {Nat. Rev. Mater.}\ }\textbf {\bibinfo {volume} {6}},\ \bibinfo
  {pages} {944} (\bibinfo {year} {2021})}\BibitemShut {NoStop}%
\bibitem [{\citenamefont {Yavilberg}\ \emph {et~al.}(2015)\citenamefont
  {Yavilberg}, \citenamefont {Ginossar},\ and\ \citenamefont
  {Grosfeld}}]{Maj_Circuit_QED2015}%
  \BibitemOpen
  \bibfield  {author} {\bibinfo {author} {\bibfnamefont {K.}~\bibnamefont
  {Yavilberg}}, \bibinfo {author} {\bibfnamefont {E.}~\bibnamefont {Ginossar}},
  \ and\ \bibinfo {author} {\bibfnamefont {E.}~\bibnamefont {Grosfeld}},\
  }\href {\doibase 10.1103/PhysRevB.92.075143} {\bibfield  {journal} {\bibinfo
  {journal} {Phys. Rev. B}\ }\textbf {\bibinfo {volume} {92}},\ \bibinfo
  {pages} {075143} (\bibinfo {year} {2015})}\BibitemShut {NoStop}%
\bibitem [{\citenamefont {Schuster}\ \emph {et~al.}(2007)\citenamefont
  {Schuster}, \citenamefont {Houck}, \citenamefont {Schreier}, \citenamefont
  {Wallraff}, \citenamefont {Gambetta}, \citenamefont {Blais}, \citenamefont
  {Frunzio}, \citenamefont {Majer}, \citenamefont {Johnson}, \citenamefont
  {Devoret}, \citenamefont {Girvin},\ and\ \citenamefont
  {Schoelkopf}}]{Dispersive_Regime_2007}%
  \BibitemOpen
  \bibfield  {author} {\bibinfo {author} {\bibfnamefont {D.~I.}\ \bibnamefont
  {Schuster}}, \bibinfo {author} {\bibfnamefont {A.~A.}\ \bibnamefont {Houck}},
  \bibinfo {author} {\bibfnamefont {J.~A.}\ \bibnamefont {Schreier}}, \bibinfo
  {author} {\bibfnamefont {A.}~\bibnamefont {Wallraff}}, \bibinfo {author}
  {\bibfnamefont {J.~M.}\ \bibnamefont {Gambetta}}, \bibinfo {author}
  {\bibfnamefont {A.}~\bibnamefont {Blais}}, \bibinfo {author} {\bibfnamefont
  {L.}~\bibnamefont {Frunzio}}, \bibinfo {author} {\bibfnamefont
  {J.}~\bibnamefont {Majer}}, \bibinfo {author} {\bibfnamefont
  {B.}~\bibnamefont {Johnson}}, \bibinfo {author} {\bibfnamefont {M.~H.}\
  \bibnamefont {Devoret}}, \bibinfo {author} {\bibfnamefont {S.~M.}\
  \bibnamefont {Girvin}}, \ and\ \bibinfo {author} {\bibfnamefont {R.~J.}\
  \bibnamefont {Schoelkopf}},\ }\href {https://doi.org/10.1038/nature05461}
  {\bibfield  {journal} {\bibinfo  {journal} {Nature}\ }\textbf {\bibinfo
  {volume} {445}},\ \bibinfo {pages} {515} (\bibinfo {year}
  {2007})}\BibitemShut {NoStop}%
\bibitem [{\citenamefont {Tinkham}(2004)}]{Tinkham_Superconductivity}%
  \BibitemOpen
  \bibfield  {author} {\bibinfo {author} {\bibfnamefont {M.}~\bibnamefont
  {Tinkham}},\ }\href@noop {} {\emph {\bibinfo {title} {Introduction to
  Superconductivity}}},\ \bibinfo {edition} {2nd}\ ed.\ (\bibinfo  {publisher}
  {Dover Publications},\ \bibinfo {address} {New York},\ \bibinfo {year}
  {2004})\BibitemShut {NoStop}%
\bibitem [{\citenamefont {de~Lange}\ \emph {et~al.}(2015)\citenamefont
  {de~Lange}, \citenamefont {van Heck}, \citenamefont {Bruno}, \citenamefont
  {van Woerkom}, \citenamefont {Geresdi}, \citenamefont {Plissard},
  \citenamefont {Bakkers}, \citenamefont {Akhmerov},\ and\ \citenamefont
  {DiCarlo}}]{Gate_Controlled_JJ_1_2015}%
  \BibitemOpen
  \bibfield  {author} {\bibinfo {author} {\bibfnamefont {G.}~\bibnamefont
  {de~Lange}}, \bibinfo {author} {\bibfnamefont {B.}~\bibnamefont {van Heck}},
  \bibinfo {author} {\bibfnamefont {A.}~\bibnamefont {Bruno}}, \bibinfo
  {author} {\bibfnamefont {D.~J.}\ \bibnamefont {van Woerkom}}, \bibinfo
  {author} {\bibfnamefont {A.}~\bibnamefont {Geresdi}}, \bibinfo {author}
  {\bibfnamefont {S.~R.}\ \bibnamefont {Plissard}}, \bibinfo {author}
  {\bibfnamefont {E.~P. A.~M.}\ \bibnamefont {Bakkers}}, \bibinfo {author}
  {\bibfnamefont {A.~R.}\ \bibnamefont {Akhmerov}}, \ and\ \bibinfo {author}
  {\bibfnamefont {L.}~\bibnamefont {DiCarlo}},\ }\href {\doibase
  10.1103/PhysRevLett.115.127002} {\bibfield  {journal} {\bibinfo  {journal}
  {Phys. Rev. Lett.}\ }\textbf {\bibinfo {volume} {115}},\ \bibinfo {pages}
  {127002} (\bibinfo {year} {2015})}\BibitemShut {NoStop}%
\bibitem [{\citenamefont {Larsen}\ \emph {et~al.}(2015)\citenamefont {Larsen},
  \citenamefont {Petersson}, \citenamefont {Kuemmeth}, \citenamefont
  {Jespersen}, \citenamefont {Krogstrup}, \citenamefont {Nyg\aa{}rd},\ and\
  \citenamefont {Marcus}}]{Gate_Controlled_JJ_2_2015}%
  \BibitemOpen
  \bibfield  {author} {\bibinfo {author} {\bibfnamefont {T.~W.}\ \bibnamefont
  {Larsen}}, \bibinfo {author} {\bibfnamefont {K.~D.}\ \bibnamefont
  {Petersson}}, \bibinfo {author} {\bibfnamefont {F.}~\bibnamefont {Kuemmeth}},
  \bibinfo {author} {\bibfnamefont {T.~S.}\ \bibnamefont {Jespersen}}, \bibinfo
  {author} {\bibfnamefont {P.}~\bibnamefont {Krogstrup}}, \bibinfo {author}
  {\bibfnamefont {J.}~\bibnamefont {Nyg\aa{}rd}}, \ and\ \bibinfo {author}
  {\bibfnamefont {C.~M.}\ \bibnamefont {Marcus}},\ }\href {\doibase
  10.1103/PhysRevLett.115.127001} {\bibfield  {journal} {\bibinfo  {journal}
  {Phys. Rev. Lett.}\ }\textbf {\bibinfo {volume} {115}},\ \bibinfo {pages}
  {127001} (\bibinfo {year} {2015})}\BibitemShut {NoStop}%
\bibitem [{\citenamefont {Clarke}\ \emph {et~al.}(2011)\citenamefont {Clarke},
  \citenamefont {Sau},\ and\ \citenamefont {Tewari}}]{SauClarkeTewari2011}%
  \BibitemOpen
  \bibfield  {author} {\bibinfo {author} {\bibfnamefont {D.~J.}\ \bibnamefont
  {Clarke}}, \bibinfo {author} {\bibfnamefont {J.~D.}\ \bibnamefont {Sau}}, \
  and\ \bibinfo {author} {\bibfnamefont {S.}~\bibnamefont {Tewari}},\ }\href
  {\doibase 10.1103/PhysRevB.84.035120} {\bibfield  {journal} {\bibinfo
  {journal} {Phys. Rev. B}\ }\textbf {\bibinfo {volume} {84}},\ \bibinfo
  {pages} {035120} (\bibinfo {year} {2011})}\BibitemShut {NoStop}%
\bibitem [{\citenamefont {van Heck}\ \emph {et~al.}(2011)\citenamefont {van
  Heck}, \citenamefont {Hassler}, \citenamefont {Akhmerov},\ and\ \citenamefont
  {Beenakker}}]{Stab_4pi_Josephson_Effect_2011}%
  \BibitemOpen
  \bibfield  {author} {\bibinfo {author} {\bibfnamefont {B.}~\bibnamefont {van
  Heck}}, \bibinfo {author} {\bibfnamefont {F.}~\bibnamefont {Hassler}},
  \bibinfo {author} {\bibfnamefont {A.~R.}\ \bibnamefont {Akhmerov}}, \ and\
  \bibinfo {author} {\bibfnamefont {C.~W.~J.}\ \bibnamefont {Beenakker}},\
  }\href {\doibase 10.1103/PhysRevB.84.180502} {\bibfield  {journal} {\bibinfo
  {journal} {Phys. Rev. B}\ }\textbf {\bibinfo {volume} {84}},\ \bibinfo
  {pages} {180502(R)} (\bibinfo {year} {2011})}\BibitemShut {NoStop}%
\bibitem [{Bra()}]{Braiding}%
  \BibitemOpen
  \href@noop {} {}\bibinfo {note} {Braids in a fermionic circuit do not
  increase the length of MZM strings in the resulting FPBC. They simply permute
  the MZMs in those strings. Hence, while their explicit implementation incurs
  some resource overhead~\cite{Majoranas_Tjunctions2011,
  Coulomb_assisted_braiding2012,
  Time_Scales_Coulomb_Block2016,Mmt_Only_QC_2008} this cost is completely
  avoided in FPBC.}\BibitemShut {Stop}%
\bibitem [{E_J()}]{E_J}%
  \BibitemOpen
  \href@noop {} {}\bibinfo {note} {One can include an extra Josephson junction
  directly connecting the bus and phase ground plates. For large Josephson
  energy, this results in the required $E_J\gg E_C$, thus allowing a smaller
  $E_{J,ab}^{\text{(off)}}$, and hence a smaller constant $c$.}\BibitemShut
  {Stop}%
\bibitem [{\citenamefont {Lisenfeld}\ \emph {et~al.}(2019)\citenamefont
  {Lisenfeld}, \citenamefont {Bilmes}, \citenamefont {Megrant}, \citenamefont
  {Barends}, \citenamefont {Kelly}, \citenamefont {Klimov}, \citenamefont
  {Weiss}, \citenamefont {Martinis},\ and\ \citenamefont
  {Ustinov}}]{Defects_Nature2019}%
  \BibitemOpen
  \bibfield  {author} {\bibinfo {author} {\bibfnamefont {J.}~\bibnamefont
  {Lisenfeld}}, \bibinfo {author} {\bibfnamefont {A.}~\bibnamefont {Bilmes}},
  \bibinfo {author} {\bibfnamefont {A.}~\bibnamefont {Megrant}}, \bibinfo
  {author} {\bibfnamefont {R.}~\bibnamefont {Barends}}, \bibinfo {author}
  {\bibfnamefont {J.}~\bibnamefont {Kelly}}, \bibinfo {author} {\bibfnamefont
  {P.}~\bibnamefont {Klimov}}, \bibinfo {author} {\bibfnamefont
  {G.}~\bibnamefont {Weiss}}, \bibinfo {author} {\bibfnamefont {J.~M.}\
  \bibnamefont {Martinis}}, \ and\ \bibinfo {author} {\bibfnamefont {A.~V.}\
  \bibnamefont {Ustinov}},\ }\href {\doibase 10.1038/s41534-019-0224-1}
  {\bibfield  {journal} {\bibinfo  {journal} {npj Quantum Inf.}\ }\textbf
  {\bibinfo {volume} {5}},\ \bibinfo {pages} {105} (\bibinfo {year}
  {2019})}\BibitemShut {NoStop}%
\bibitem [{\citenamefont {de~Graaf}\ \emph {et~al.}(2020)\citenamefont
  {de~Graaf}, \citenamefont {Faoro}, \citenamefont {Ioffe}, \citenamefont
  {Mahashabde}, \citenamefont {Burnett}, \citenamefont {Lindström},
  \citenamefont {Kubatkin}, \citenamefont {Danilov},\ and\ \citenamefont
  {Tzalenchuk}}]{Trapped_QPs_transmons2020}%
  \BibitemOpen
  \bibfield  {author} {\bibinfo {author} {\bibfnamefont {S.~E.}\ \bibnamefont
  {de~Graaf}}, \bibinfo {author} {\bibfnamefont {L.}~\bibnamefont {Faoro}},
  \bibinfo {author} {\bibfnamefont {L.~B.}\ \bibnamefont {Ioffe}}, \bibinfo
  {author} {\bibfnamefont {S.}~\bibnamefont {Mahashabde}}, \bibinfo {author}
  {\bibfnamefont {J.~J.}\ \bibnamefont {Burnett}}, \bibinfo {author}
  {\bibfnamefont {T.}~\bibnamefont {Lindström}}, \bibinfo {author}
  {\bibfnamefont {S.~E.}\ \bibnamefont {Kubatkin}}, \bibinfo {author}
  {\bibfnamefont {A.~V.}\ \bibnamefont {Danilov}}, \ and\ \bibinfo {author}
  {\bibfnamefont {A.~Y.}\ \bibnamefont {Tzalenchuk}},\ }\href {\doibase
  10.1126/sciadv.abc5055} {\bibfield  {journal} {\bibinfo  {journal} {Sci.
  Adv.}\ }\textbf {\bibinfo {volume} {6}},\ \bibinfo {pages} {eabc5055}
  (\bibinfo {year} {2020})}\BibitemShut {NoStop}%
\bibitem [{\citenamefont
  {Siddiqi}(2021)}]{Review_Decoherence_Supercond_Qubits2021}%
  \BibitemOpen
  \bibfield  {author} {\bibinfo {author} {\bibfnamefont {I.}~\bibnamefont
  {Siddiqi}},\ }\href {\doibase 10.1038/s41578-021-00370-4} {\bibfield
  {journal} {\bibinfo  {journal} {Nat. Rev. Mater.}\ }\textbf {\bibinfo
  {volume} {6}},\ \bibinfo {pages} {875} (\bibinfo {year} {2021})}\BibitemShut
  {NoStop}%
\bibitem [{\citenamefont {de~Leon}\ \emph {et~al.}(2021)\citenamefont
  {de~Leon}, \citenamefont {Itoh}, \citenamefont {Kim}, \citenamefont {Mehta},
  \citenamefont {Northup}, \citenamefont {Paik}, \citenamefont {Palmer},
  \citenamefont {Samarth}, \citenamefont {Sangtawesin},\ and\ \citenamefont
  {Steuerman}}]{Materials_Challenge_QC2021}%
  \BibitemOpen
  \bibfield  {author} {\bibinfo {author} {\bibfnamefont {N.~P.}\ \bibnamefont
  {de~Leon}}, \bibinfo {author} {\bibfnamefont {K.~M.}\ \bibnamefont {Itoh}},
  \bibinfo {author} {\bibfnamefont {D.}~\bibnamefont {Kim}}, \bibinfo {author}
  {\bibfnamefont {K.~K.}\ \bibnamefont {Mehta}}, \bibinfo {author}
  {\bibfnamefont {T.~E.}\ \bibnamefont {Northup}}, \bibinfo {author}
  {\bibfnamefont {H.}~\bibnamefont {Paik}}, \bibinfo {author} {\bibfnamefont
  {B.~S.}\ \bibnamefont {Palmer}}, \bibinfo {author} {\bibfnamefont
  {N.}~\bibnamefont {Samarth}}, \bibinfo {author} {\bibfnamefont
  {S.}~\bibnamefont {Sangtawesin}}, \ and\ \bibinfo {author} {\bibfnamefont
  {D.~W.}\ \bibnamefont {Steuerman}},\ }\href {\doibase
  10.1126/science.abb2823} {\bibfield  {journal} {\bibinfo  {journal}
  {Science}\ }\textbf {\bibinfo {volume} {372}},\ \bibinfo {pages} {eabb2823}
  (\bibinfo {year} {2021})}\BibitemShut {NoStop}%
\bibitem [{\citenamefont {Bilmes}\ \emph {et~al.}(2020)\citenamefont {Bilmes},
  \citenamefont {Megrant}, \citenamefont {Klimov}, \citenamefont {Weiss},
  \citenamefont {Martinis}, \citenamefont {Ustinov},\ and\ \citenamefont
  {Lisenfeld}}]{Defect_location2020}%
  \BibitemOpen
  \bibfield  {author} {\bibinfo {author} {\bibfnamefont {A.}~\bibnamefont
  {Bilmes}}, \bibinfo {author} {\bibfnamefont {A.}~\bibnamefont {Megrant}},
  \bibinfo {author} {\bibfnamefont {P.}~\bibnamefont {Klimov}}, \bibinfo
  {author} {\bibfnamefont {G.}~\bibnamefont {Weiss}}, \bibinfo {author}
  {\bibfnamefont {J.~M.}\ \bibnamefont {Martinis}}, \bibinfo {author}
  {\bibfnamefont {A.~V.}\ \bibnamefont {Ustinov}}, \ and\ \bibinfo {author}
  {\bibfnamefont {J.}~\bibnamefont {Lisenfeld}},\ }\href {\doibase
  10.1038/s41598-020-59749-y} {\bibfield  {journal} {\bibinfo  {journal} {Sci.
  Rep.}\ }\textbf {\bibinfo {volume} {10}},\ \bibinfo {pages} {3090} (\bibinfo
  {year} {2020})}\BibitemShut {NoStop}%
\bibitem [{\citenamefont {Andersson}\ \emph {et~al.}(2021)\citenamefont
  {Andersson}, \citenamefont {Bilobran}, \citenamefont {Scigliuzzo},
  \citenamefont {de~Lima}, \citenamefont {Cole},\ and\ \citenamefont
  {Delsing}}]{TLS_Ensemble2021}%
  \BibitemOpen
  \bibfield  {author} {\bibinfo {author} {\bibfnamefont {G.}~\bibnamefont
  {Andersson}}, \bibinfo {author} {\bibfnamefont {A.~L.~O.}\ \bibnamefont
  {Bilobran}}, \bibinfo {author} {\bibfnamefont {M.}~\bibnamefont
  {Scigliuzzo}}, \bibinfo {author} {\bibfnamefont {M.~M.}\ \bibnamefont
  {de~Lima}}, \bibinfo {author} {\bibfnamefont {J.~H.}\ \bibnamefont {Cole}}, \
  and\ \bibinfo {author} {\bibfnamefont {P.}~\bibnamefont {Delsing}},\ }\href
  {\doibase 10.1038/s41534-020-00348-0} {\bibfield  {journal} {\bibinfo
  {journal} {npj Quantum Inf.}\ }\textbf {\bibinfo {volume} {7}},\ \bibinfo
  {pages} {15} (\bibinfo {year} {2021})}\BibitemShut {NoStop}%
\bibitem [{\citenamefont {Dunsworth}\ \emph {et~al.}(2017)\citenamefont
  {Dunsworth}, \citenamefont {Megrant}, \citenamefont {Quintana}, \citenamefont
  {Chen}, \citenamefont {Barends}, \citenamefont {Burkett}, \citenamefont
  {Foxen}, \citenamefont {Chen}, \citenamefont {Chiaro}, \citenamefont {Fowler}
  \emph {et~al.}}]{Reducing_JJ_Capacitive_Loss2017}%
  \BibitemOpen
  \bibfield  {author} {\bibinfo {author} {\bibfnamefont {A.}~\bibnamefont
  {Dunsworth}}, \bibinfo {author} {\bibfnamefont {A.}~\bibnamefont {Megrant}},
  \bibinfo {author} {\bibfnamefont {C.}~\bibnamefont {Quintana}}, \bibinfo
  {author} {\bibfnamefont {Z.}~\bibnamefont {Chen}}, \bibinfo {author}
  {\bibfnamefont {R.}~\bibnamefont {Barends}}, \bibinfo {author} {\bibfnamefont
  {B.}~\bibnamefont {Burkett}}, \bibinfo {author} {\bibfnamefont
  {B.}~\bibnamefont {Foxen}}, \bibinfo {author} {\bibfnamefont
  {Y.}~\bibnamefont {Chen}}, \bibinfo {author} {\bibfnamefont {B.}~\bibnamefont
  {Chiaro}}, \bibinfo {author} {\bibfnamefont {A.}~\bibnamefont {Fowler}},
  \emph {et~al.},\ }\href {https://doi.org/10.1063/1.4993577} {\bibfield
  {journal} {\bibinfo  {journal} {Appl. Phys. Lett.}\ }\textbf {\bibinfo
  {volume} {111}},\ \bibinfo {pages} {022601} (\bibinfo {year}
  {2017})}\BibitemShut {NoStop}%
\bibitem [{\citenamefont {Place}\ \emph {et~al.}(2021)\citenamefont {Place},
  \citenamefont {Rodgers}, \citenamefont {Mundada}, \citenamefont {Smitham},
  \citenamefont {Fitzpatrick}, \citenamefont {Leng}, \citenamefont {Premkumar},
  \citenamefont {Bryon}, \citenamefont {Vrajitoarea}, \citenamefont {Sussman}
  \emph {et~al.}}]{Coherence_Times_Transmon2021}%
  \BibitemOpen
  \bibfield  {author} {\bibinfo {author} {\bibfnamefont {A.~P.~M.}\
  \bibnamefont {Place}}, \bibinfo {author} {\bibfnamefont {L.~V.~H.}\
  \bibnamefont {Rodgers}}, \bibinfo {author} {\bibfnamefont {P.}~\bibnamefont
  {Mundada}}, \bibinfo {author} {\bibfnamefont {B.~M.}\ \bibnamefont
  {Smitham}}, \bibinfo {author} {\bibfnamefont {M.}~\bibnamefont
  {Fitzpatrick}}, \bibinfo {author} {\bibfnamefont {Z.}~\bibnamefont {Leng}},
  \bibinfo {author} {\bibfnamefont {A.}~\bibnamefont {Premkumar}}, \bibinfo
  {author} {\bibfnamefont {J.}~\bibnamefont {Bryon}}, \bibinfo {author}
  {\bibfnamefont {A.}~\bibnamefont {Vrajitoarea}}, \bibinfo {author}
  {\bibfnamefont {S.}~\bibnamefont {Sussman}},  \emph {et~al.},\ }\href
  {\doibase 10.1038/s41467-021-22030-5} {\bibfield  {journal} {\bibinfo
  {journal} {Nat. Commun.}\ }\textbf {\bibinfo {volume} {12}},\ \bibinfo
  {pages} {1779} (\bibinfo {year} {2021})}\BibitemShut {NoStop}%
\bibitem [{\citenamefont {Osman}\ \emph {et~al.}(2021)\citenamefont {Osman},
  \citenamefont {Simon}, \citenamefont {Bengtsson}, \citenamefont {Kosen},
  \citenamefont {Krantz}, \citenamefont {P.~Lozano}, \citenamefont
  {Scigliuzzo}, \citenamefont {Delsing}, \citenamefont {Bylander},\ and\
  \citenamefont {Fadavi~Roudsari}}]{Better_JJ_Fab2021}%
  \BibitemOpen
  \bibfield  {author} {\bibinfo {author} {\bibfnamefont {A.}~\bibnamefont
  {Osman}}, \bibinfo {author} {\bibfnamefont {J.}~\bibnamefont {Simon}},
  \bibinfo {author} {\bibfnamefont {A.}~\bibnamefont {Bengtsson}}, \bibinfo
  {author} {\bibfnamefont {S.}~\bibnamefont {Kosen}}, \bibinfo {author}
  {\bibfnamefont {P.}~\bibnamefont {Krantz}}, \bibinfo {author} {\bibfnamefont
  {D.}~\bibnamefont {P.~Lozano}}, \bibinfo {author} {\bibfnamefont
  {M.}~\bibnamefont {Scigliuzzo}}, \bibinfo {author} {\bibfnamefont
  {P.}~\bibnamefont {Delsing}}, \bibinfo {author} {\bibfnamefont
  {J.}~\bibnamefont {Bylander}}, \ and\ \bibinfo {author} {\bibfnamefont
  {A.}~\bibnamefont {Fadavi~Roudsari}},\ }\href {\doibase 10.1063/5.0037093}
  {\bibfield  {journal} {\bibinfo  {journal} {Appl. Phys. Lett.}\ }\textbf
  {\bibinfo {volume} {118}},\ \bibinfo {pages} {064002} (\bibinfo {year}
  {2021})}\BibitemShut {NoStop}%
\bibitem [{\citenamefont {Bilmes}\ \emph {et~al.}(2021)\citenamefont {Bilmes},
  \citenamefont {Händel}, \citenamefont {Volosheniuk}, \citenamefont
  {Ustinov},\ and\ \citenamefont {Lisenfeld}}]{Bandaged_JJs2021}%
  \BibitemOpen
  \bibfield  {author} {\bibinfo {author} {\bibfnamefont {A.}~\bibnamefont
  {Bilmes}}, \bibinfo {author} {\bibfnamefont {A.~K.}\ \bibnamefont {Händel}},
  \bibinfo {author} {\bibfnamefont {S.}~\bibnamefont {Volosheniuk}}, \bibinfo
  {author} {\bibfnamefont {A.~V.}\ \bibnamefont {Ustinov}}, \ and\ \bibinfo
  {author} {\bibfnamefont {J.}~\bibnamefont {Lisenfeld}},\ }\href {\doibase
  10.1088/1361-6668/ac2a6d} {\bibfield  {journal} {\bibinfo  {journal}
  {Supercond. Sci. Technol.}\ }\textbf {\bibinfo {volume} {34}},\ \bibinfo
  {pages} {125011} (\bibinfo {year} {2021})}\BibitemShut {NoStop}%
\bibitem [{\citenamefont {Knapp}\ \emph {et~al.}(2018)\citenamefont {Knapp},
  \citenamefont {Karzig}, \citenamefont {Lutchyn},\ and\ \citenamefont
  {Nayak}}]{MZM_Coherence_Times2018}%
  \BibitemOpen
  \bibfield  {author} {\bibinfo {author} {\bibfnamefont {C.}~\bibnamefont
  {Knapp}}, \bibinfo {author} {\bibfnamefont {T.}~\bibnamefont {Karzig}},
  \bibinfo {author} {\bibfnamefont {R.~M.}\ \bibnamefont {Lutchyn}}, \ and\
  \bibinfo {author} {\bibfnamefont {C.}~\bibnamefont {Nayak}},\ }\href
  {\doibase 10.1103/PhysRevB.97.125404} {\bibfield  {journal} {\bibinfo
  {journal} {Phys. Rev. B}\ }\textbf {\bibinfo {volume} {97}},\ \bibinfo
  {pages} {125404} (\bibinfo {year} {2018})}\BibitemShut {NoStop}%
\bibitem [{\citenamefont {Karzig}\ \emph {et~al.}(2021)\citenamefont {Karzig},
  \citenamefont {Cole},\ and\ \citenamefont
  {Pikulin}}]{Karzig_QPP_Maj_Qubits_2021}%
  \BibitemOpen
  \bibfield  {author} {\bibinfo {author} {\bibfnamefont {T.}~\bibnamefont
  {Karzig}}, \bibinfo {author} {\bibfnamefont {W.~S.}\ \bibnamefont {Cole}}, \
  and\ \bibinfo {author} {\bibfnamefont {D.~I.}\ \bibnamefont {Pikulin}},\
  }\href {http://dx.doi.org/10.1103/PhysRevLett.126.057702} {\bibfield
  {journal} {\bibinfo  {journal} {Phys. Rev. Lett.}\ }\textbf {\bibinfo
  {volume} {126}},\ \bibinfo {pages} {057702} (\bibinfo {year}
  {2021})}\BibitemShut {NoStop}%
\bibitem [{\citenamefont {Riwar}\ \emph {et~al.}(2016)\citenamefont {Riwar},
  \citenamefont {Hosseinkhani}, \citenamefont {Burkhart}, \citenamefont {Gao},
  \citenamefont {Schoelkopf}, \citenamefont {Glazman},\ and\ \citenamefont
  {Catelani}}]{QP_Traps_2016}%
  \BibitemOpen
  \bibfield  {author} {\bibinfo {author} {\bibfnamefont {R.-P.}\ \bibnamefont
  {Riwar}}, \bibinfo {author} {\bibfnamefont {A.}~\bibnamefont {Hosseinkhani}},
  \bibinfo {author} {\bibfnamefont {L.~D.}\ \bibnamefont {Burkhart}}, \bibinfo
  {author} {\bibfnamefont {Y.~Y.}\ \bibnamefont {Gao}}, \bibinfo {author}
  {\bibfnamefont {R.~J.}\ \bibnamefont {Schoelkopf}}, \bibinfo {author}
  {\bibfnamefont {L.~I.}\ \bibnamefont {Glazman}}, \ and\ \bibinfo {author}
  {\bibfnamefont {G.}~\bibnamefont {Catelani}},\ }\href {\doibase
  10.1103/PhysRevB.94.104516} {\bibfield  {journal} {\bibinfo  {journal} {Phys.
  Rev. B}\ }\textbf {\bibinfo {volume} {94}},\ \bibinfo {pages} {104516}
  (\bibinfo {year} {2016})}\BibitemShut {NoStop}%
\bibitem [{\citenamefont {Taupin}\ \emph {et~al.}(2016)\citenamefont {Taupin},
  \citenamefont {Khaymovich}, \citenamefont {Meschke}, \citenamefont
  {Mel'nikov},\ and\ \citenamefont {Pekola}}]{QP_Traps_2016_2}%
  \BibitemOpen
  \bibfield  {author} {\bibinfo {author} {\bibfnamefont {M.}~\bibnamefont
  {Taupin}}, \bibinfo {author} {\bibfnamefont {I.~M.}\ \bibnamefont
  {Khaymovich}}, \bibinfo {author} {\bibfnamefont {M.}~\bibnamefont {Meschke}},
  \bibinfo {author} {\bibfnamefont {A.~S.}\ \bibnamefont {Mel'nikov}}, \ and\
  \bibinfo {author} {\bibfnamefont {J.~P.}\ \bibnamefont {Pekola}},\ }\href
  {\doibase 10.1038/ncomms10977} {\bibfield  {journal} {\bibinfo  {journal}
  {Nat. Commun.}\ }\textbf {\bibinfo {volume} {7}},\ \bibinfo {pages} {10977}
  (\bibinfo {year} {2016})}\BibitemShut {NoStop}%
\bibitem [{\citenamefont {Tornberg}\ and\ \citenamefont
  {Johansson}(2010)}]{Qubit_Parity_Mmt2010}%
  \BibitemOpen
  \bibfield  {author} {\bibinfo {author} {\bibfnamefont {L.}~\bibnamefont
  {Tornberg}}\ and\ \bibinfo {author} {\bibfnamefont {G.}~\bibnamefont
  {Johansson}},\ }\href {\doibase 10.1103/PhysRevA.82.012329} {\bibfield
  {journal} {\bibinfo  {journal} {Phys. Rev. A}\ }\textbf {\bibinfo {volume}
  {82}},\ \bibinfo {pages} {012329} (\bibinfo {year} {2010})}\BibitemShut
  {NoStop}%
\bibitem [{\citenamefont {Lalumi\`ere}\ \emph {et~al.}(2010)\citenamefont
  {Lalumi\`ere}, \citenamefont {Gambetta},\ and\ \citenamefont
  {Blais}}]{Parity_Mmt_2_2010}%
  \BibitemOpen
  \bibfield  {author} {\bibinfo {author} {\bibfnamefont {K.}~\bibnamefont
  {Lalumi\`ere}}, \bibinfo {author} {\bibfnamefont {J.~M.}\ \bibnamefont
  {Gambetta}}, \ and\ \bibinfo {author} {\bibfnamefont {A.}~\bibnamefont
  {Blais}},\ }\href {\doibase 10.1103/PhysRevA.81.040301} {\bibfield  {journal}
  {\bibinfo  {journal} {Phys. Rev. A}\ }\textbf {\bibinfo {volume} {81}},\
  \bibinfo {pages} {040301(R)} (\bibinfo {year} {2010})}\BibitemShut {NoStop}%
\bibitem [{\citenamefont {Rist{\`e}}\ \emph {et~al.}(2013)\citenamefont
  {Rist{\`e}}, \citenamefont {Dukalski}, \citenamefont {Watson}, \citenamefont
  {de~Lange}, \citenamefont {Tiggelman}, \citenamefont {Blanter}, \citenamefont
  {Lehnert}, \citenamefont {Schouten},\ and\ \citenamefont
  {DiCarlo}}]{Parity_Mmt_2013}%
  \BibitemOpen
  \bibfield  {author} {\bibinfo {author} {\bibfnamefont {D.}~\bibnamefont
  {Rist{\`e}}}, \bibinfo {author} {\bibfnamefont {M.}~\bibnamefont {Dukalski}},
  \bibinfo {author} {\bibfnamefont {C.~A.}\ \bibnamefont {Watson}}, \bibinfo
  {author} {\bibfnamefont {G.}~\bibnamefont {de~Lange}}, \bibinfo {author}
  {\bibfnamefont {M.~J.}\ \bibnamefont {Tiggelman}}, \bibinfo {author}
  {\bibfnamefont {Y.~M.}\ \bibnamefont {Blanter}}, \bibinfo {author}
  {\bibfnamefont {K.~W.}\ \bibnamefont {Lehnert}}, \bibinfo {author}
  {\bibfnamefont {R.~N.}\ \bibnamefont {Schouten}}, \ and\ \bibinfo {author}
  {\bibfnamefont {L.}~\bibnamefont {DiCarlo}},\ }\href {\doibase
  10.1038/nature12513} {\bibfield  {journal} {\bibinfo  {journal} {Nature}\
  }\textbf {\bibinfo {volume} {502}},\ \bibinfo {pages} {350} (\bibinfo {year}
  {2013})}\BibitemShut {NoStop}%
\bibitem [{\citenamefont {Royer}\ \emph {et~al.}(2018)\citenamefont {Royer},
  \citenamefont {Puri},\ and\ \citenamefont {Blais}}]{Qubit_Parity_Mmt2018}%
  \BibitemOpen
  \bibfield  {author} {\bibinfo {author} {\bibfnamefont {B.}~\bibnamefont
  {Royer}}, \bibinfo {author} {\bibfnamefont {S.}~\bibnamefont {Puri}}, \ and\
  \bibinfo {author} {\bibfnamefont {A.}~\bibnamefont {Blais}},\ }\href
  {\doibase 10.1126/sciadv.aau1695} {\bibfield  {journal} {\bibinfo  {journal}
  {Sci. Adv.}\ }\textbf {\bibinfo {volume} {4}},\ \bibinfo {pages} {eaau1695}
  (\bibinfo {year} {2018})}\BibitemShut {NoStop}%
\bibitem [{\citenamefont {Akhmerov}(2010)}]{Akhmerov2010}%
  \BibitemOpen
  \bibfield  {author} {\bibinfo {author} {\bibfnamefont {A.~R.}\ \bibnamefont
  {Akhmerov}},\ }\href {\doibase 10.1103/PhysRevB.82.020509} {\bibfield
  {journal} {\bibinfo  {journal} {Phys. Rev. B}\ }\textbf {\bibinfo {volume}
  {82}},\ \bibinfo {pages} {020509(R)} (\bibinfo {year} {2010})}\BibitemShut
  {NoStop}%
\bibitem [{\citenamefont {Walter}\ \emph {et~al.}(2017)\citenamefont {Walter},
  \citenamefont {Kurpiers}, \citenamefont {Gasparinetti}, \citenamefont
  {Magnard}, \citenamefont {Poto\ifmmode~\check{c}\else \v{c}\fi{}nik},
  \citenamefont {Salath\'e}, \citenamefont {Pechal}, \citenamefont {Mondal},
  \citenamefont {Oppliger}, \citenamefont {Eichler},\ and\ \citenamefont
  {Wallraff}}]{Better_Transmon_Mmt_1_2017}%
  \BibitemOpen
  \bibfield  {author} {\bibinfo {author} {\bibfnamefont {T.}~\bibnamefont
  {Walter}}, \bibinfo {author} {\bibfnamefont {P.}~\bibnamefont {Kurpiers}},
  \bibinfo {author} {\bibfnamefont {S.}~\bibnamefont {Gasparinetti}}, \bibinfo
  {author} {\bibfnamefont {P.}~\bibnamefont {Magnard}}, \bibinfo {author}
  {\bibfnamefont {A.}~\bibnamefont {Poto\ifmmode~\check{c}\else
  \v{c}\fi{}nik}}, \bibinfo {author} {\bibfnamefont {Y.}~\bibnamefont
  {Salath\'e}}, \bibinfo {author} {\bibfnamefont {M.}~\bibnamefont {Pechal}},
  \bibinfo {author} {\bibfnamefont {M.}~\bibnamefont {Mondal}}, \bibinfo
  {author} {\bibfnamefont {M.}~\bibnamefont {Oppliger}}, \bibinfo {author}
  {\bibfnamefont {C.}~\bibnamefont {Eichler}}, \ and\ \bibinfo {author}
  {\bibfnamefont {A.}~\bibnamefont {Wallraff}},\ }\href {\doibase
  10.1103/PhysRevApplied.7.054020} {\bibfield  {journal} {\bibinfo  {journal}
  {Phys. Rev. Appl.}\ }\textbf {\bibinfo {volume} {7}},\ \bibinfo {pages}
  {054020} (\bibinfo {year} {2017})}\BibitemShut {NoStop}%
\bibitem [{\citenamefont {Dassonneville}\ \emph {et~al.}(2020)\citenamefont
  {Dassonneville}, \citenamefont {Ramos}, \citenamefont {Milchakov},
  \citenamefont {Planat}, \citenamefont {Dumur}, \citenamefont {Foroughi},
  \citenamefont {Puertas}, \citenamefont {Leger}, \citenamefont {Bharadwaj},
  \citenamefont {Delaforce} \emph {et~al.}}]{Better_transmon_mmts2020}%
  \BibitemOpen
  \bibfield  {author} {\bibinfo {author} {\bibfnamefont {R.}~\bibnamefont
  {Dassonneville}}, \bibinfo {author} {\bibfnamefont {T.}~\bibnamefont
  {Ramos}}, \bibinfo {author} {\bibfnamefont {V.}~\bibnamefont {Milchakov}},
  \bibinfo {author} {\bibfnamefont {L.}~\bibnamefont {Planat}}, \bibinfo
  {author} {\bibfnamefont {E.}~\bibnamefont {Dumur}}, \bibinfo {author}
  {\bibfnamefont {F.}~\bibnamefont {Foroughi}}, \bibinfo {author}
  {\bibfnamefont {J.}~\bibnamefont {Puertas}}, \bibinfo {author} {\bibfnamefont
  {S.}~\bibnamefont {Leger}}, \bibinfo {author} {\bibfnamefont
  {K.}~\bibnamefont {Bharadwaj}}, \bibinfo {author} {\bibfnamefont
  {J.}~\bibnamefont {Delaforce}},  \emph {et~al.},\ }\href {\doibase
  10.1103/PhysRevX.10.011045} {\bibfield  {journal} {\bibinfo  {journal} {Phys.
  Rev. X}\ }\textbf {\bibinfo {volume} {10}},\ \bibinfo {pages} {011045}
  (\bibinfo {year} {2020})}\BibitemShut {NoStop}%
\bibitem [{\citenamefont {Helmer}\ \emph {et~al.}(2009)\citenamefont {Helmer},
  \citenamefont {Mariantoni}, \citenamefont {Fowler}, \citenamefont {von
  Delft}, \citenamefont {Solano},\ and\ \citenamefont
  {Marquardt}}]{2D_Cavity_grid_2009}%
  \BibitemOpen
  \bibfield  {author} {\bibinfo {author} {\bibfnamefont {F.}~\bibnamefont
  {Helmer}}, \bibinfo {author} {\bibfnamefont {M.}~\bibnamefont {Mariantoni}},
  \bibinfo {author} {\bibfnamefont {A.~G.}\ \bibnamefont {Fowler}}, \bibinfo
  {author} {\bibfnamefont {J.}~\bibnamefont {von Delft}}, \bibinfo {author}
  {\bibfnamefont {E.}~\bibnamefont {Solano}}, \ and\ \bibinfo {author}
  {\bibfnamefont {F.}~\bibnamefont {Marquardt}},\ }\href {\doibase
  10.1209/0295-5075/85/50007} {\bibfield  {journal} {\bibinfo  {journal}
  {{EPL}}\ }\textbf {\bibinfo {volume} {85}},\ \bibinfo {pages} {50007}
  (\bibinfo {year} {2009})}\BibitemShut {NoStop}%
\bibitem [{\citenamefont {Arute}\ \emph {et~al.}(2019)\citenamefont {Arute}
  \emph {et~al.}}]{Quantum_Supremacy2019}%
  \BibitemOpen
  \bibfield  {author} {\bibinfo {author} {\bibfnamefont {F.}~\bibnamefont
  {Arute}} \emph {et~al.},\ }\href {\doibase 10.1038/s41586-019-1666-5}
  {\bibfield  {journal} {\bibinfo  {journal} {Nature}\ }\textbf {\bibinfo
  {volume} {574}},\ \bibinfo {pages} {505} (\bibinfo {year}
  {2019})}\BibitemShut {NoStop}%
\bibitem [{\citenamefont {Patterson}\ \emph {et~al.}(2019)\citenamefont
  {Patterson}, \citenamefont {Rahamim}, \citenamefont {Tsunoda}, \citenamefont
  {Spring}, \citenamefont {Jebari}, \citenamefont {Ratter}, \citenamefont
  {Mergenthaler}, \citenamefont {Tancredi}, \citenamefont {Vlastakis},
  \citenamefont {Esposito},\ and\ \citenamefont
  {Leek}}]{Calibrating_2_qubit_gate2019}%
  \BibitemOpen
  \bibfield  {author} {\bibinfo {author} {\bibfnamefont {A.~D.}~\bibnamefont
  {Patterson}}, \bibinfo {author} {\bibfnamefont {J.}~\bibnamefont {Rahamim}},
  \bibinfo {author} {\bibfnamefont {T.}~\bibnamefont {Tsunoda}}, \bibinfo
  {author} {\bibfnamefont {P.~A.}~\bibnamefont {Spring}}, \bibinfo {author}
  {\bibfnamefont {S.}~\bibnamefont {Jebari}}, \bibinfo {author} {\bibfnamefont
  {K.}~\bibnamefont {Ratter}}, \bibinfo {author} {\bibfnamefont
  {M.}~\bibnamefont {Mergenthaler}}, \bibinfo {author} {\bibfnamefont
  {G.}~\bibnamefont {Tancredi}}, \bibinfo {author} {\bibfnamefont
  {B.}~\bibnamefont {Vlastakis}}, \bibinfo {author} {\bibfnamefont
  {M.}~\bibnamefont {Esposito}}, \ and\ \bibinfo {author} {\bibfnamefont
  {P.~J.}~\bibnamefont {Leek}},\ }\href {\doibase
  10.1103/PhysRevApplied.12.064013} {\bibfield  {journal} {\bibinfo  {journal}
  {Phys. Rev. Appl.}\ }\textbf {\bibinfo {volume} {12}},\ \bibinfo {pages}
  {064013} (\bibinfo {year} {2019})}\BibitemShut {NoStop}%
\bibitem [{\citenamefont {Gao}\ \emph {et~al.}(2021)\citenamefont {Gao},
  \citenamefont {Rol}, \citenamefont {Touzard},\ and\ \citenamefont
  {Wang}}]{Calibrating_devices_review2021}%
  \BibitemOpen
  \bibfield  {author} {\bibinfo {author} {\bibfnamefont {Y.~Y.}\ \bibnamefont
  {Gao}}, \bibinfo {author} {\bibfnamefont {M.~A.}\ \bibnamefont {Rol}},
  \bibinfo {author} {\bibfnamefont {S.}~\bibnamefont {Touzard}}, \ and\
  \bibinfo {author} {\bibfnamefont {C.}~\bibnamefont {Wang}},\ }\href {\doibase
  10.1103/PRXQuantum.2.040202} {\bibfield  {journal} {\bibinfo  {journal} {PRX
  Quantum}\ }\textbf {\bibinfo {volume} {2}},\ \bibinfo {pages} {040202}
  (\bibinfo {year} {2021})}\BibitemShut {NoStop}%
\bibitem [{\citenamefont {Zhao}\ \emph {et~al.}(2021)\citenamefont {Zhao},
  \citenamefont {Lan}, \citenamefont {Xu}, \citenamefont {Xue}, \citenamefont
  {Blank}, \citenamefont {Tan}, \citenamefont {Yu},\ and\ \citenamefont
  {Yu}}]{Suppressing_ZZ_int_scaling2021}%
  \BibitemOpen
  \bibfield  {author} {\bibinfo {author} {\bibfnamefont {P.}~\bibnamefont
  {Zhao}}, \bibinfo {author} {\bibfnamefont {D.}~\bibnamefont {Lan}}, \bibinfo
  {author} {\bibfnamefont {P.}~\bibnamefont {Xu}}, \bibinfo {author}
  {\bibfnamefont {G.}~\bibnamefont {Xue}}, \bibinfo {author} {\bibfnamefont
  {M.}~\bibnamefont {Blank}}, \bibinfo {author} {\bibfnamefont
  {X.}~\bibnamefont {Tan}}, \bibinfo {author} {\bibfnamefont {H.}~\bibnamefont
  {Yu}}, \ and\ \bibinfo {author} {\bibfnamefont {Y.}~\bibnamefont {Yu}},\
  }\href {\doibase 10.1103/PhysRevApplied.16.024037} {\bibfield  {journal}
  {\bibinfo  {journal} {Phys. Rev. Appl.}\ }\textbf {\bibinfo {volume} {16}},\
  \bibinfo {pages} {024037} (\bibinfo {year} {2021})}\BibitemShut {NoStop}%
\bibitem [{\citenamefont {Bravyi}\ \emph {et~al.}(2010)\citenamefont {Bravyi},
  \citenamefont {Terhal},\ and\ \citenamefont {Leemhuis}}]{Maj_Ferm_Codes}%
  \BibitemOpen
  \bibfield  {author} {\bibinfo {author} {\bibfnamefont {S.}~\bibnamefont
  {Bravyi}}, \bibinfo {author} {\bibfnamefont {B.~M.}\ \bibnamefont {Terhal}},
  \ and\ \bibinfo {author} {\bibfnamefont {B.}~\bibnamefont {Leemhuis}},\
  }\href {\doibase 10.1088/1367-2630/12/8/083039} {\bibfield  {journal}
  {\bibinfo  {journal} {New J. Phys.}\ }\textbf {\bibinfo {volume} {12}},\
  \bibinfo {pages} {083039} (\bibinfo {year} {2010})}\BibitemShut {NoStop}%
\bibitem [{\citenamefont {Goldstein}\ and\ \citenamefont
  {Chamon}(2012)}]{Goldstein2012}%
  \BibitemOpen
  \bibfield  {author} {\bibinfo {author} {\bibfnamefont {G.}~\bibnamefont
  {Goldstein}}\ and\ \bibinfo {author} {\bibfnamefont {C.}~\bibnamefont
  {Chamon}},\ }\href {\doibase 10.1103/PhysRevB.86.115122} {\bibfield
  {journal} {\bibinfo  {journal} {Phys. Rev. B}\ }\textbf {\bibinfo {volume}
  {86}},\ \bibinfo {pages} {115122} (\bibinfo {year} {2012})}\BibitemShut
  {NoStop}%
\bibitem [{\citenamefont {Behrends}\ and\ \citenamefont
  {B\'eri}(2020)}]{Behrends2020}%
  \BibitemOpen
  \bibfield  {author} {\bibinfo {author} {\bibfnamefont {J.}~\bibnamefont
  {Behrends}}\ and\ \bibinfo {author} {\bibfnamefont {B.}~\bibnamefont
  {B\'eri}},\ }\href {\doibase 10.1103/PhysRevLett.124.236804} {\bibfield
  {journal} {\bibinfo  {journal} {Phys. Rev. Lett.}\ }\textbf {\bibinfo
  {volume} {124}},\ \bibinfo {pages} {236804} (\bibinfo {year}
  {2020})}\BibitemShut {NoStop}%
\bibitem [{\citenamefont {Plugge}\ \emph {et~al.}(2017)\citenamefont {Plugge},
  \citenamefont {Rasmussen}, \citenamefont {Egger},\ and\ \citenamefont
  {Flensberg}}]{Maj_Box_Qubits2017}%
  \BibitemOpen
  \bibfield  {author} {\bibinfo {author} {\bibfnamefont {S.}~\bibnamefont
  {Plugge}}, \bibinfo {author} {\bibfnamefont {A.}~\bibnamefont {Rasmussen}},
  \bibinfo {author} {\bibfnamefont {R.}~\bibnamefont {Egger}}, \ and\ \bibinfo
  {author} {\bibfnamefont {K.}~\bibnamefont {Flensberg}},\ }\href {\doibase
  10.1088/1367-2630/aa54e1} {\bibfield  {journal} {\bibinfo  {journal} {New J.
  Phys.}\ }\textbf {\bibinfo {volume} {19}},\ \bibinfo {pages} {012001}
  (\bibinfo {year} {2017})}\BibitemShut {NoStop}%
\bibitem [{\citenamefont {Steiner}\ and\ \citenamefont {von
  Oppen}(2020)}]{Quantum_Dot_Readout2020}%
  \BibitemOpen
  \bibfield  {author} {\bibinfo {author} {\bibfnamefont {J.~F.}\ \bibnamefont
  {Steiner}}\ and\ \bibinfo {author} {\bibfnamefont {F.}~\bibnamefont {von
  Oppen}},\ }\href {\doibase 10.1103/PhysRevResearch.2.033255} {\bibfield
  {journal} {\bibinfo  {journal} {Phys. Rev. Res.}\ }\textbf {\bibinfo {volume}
  {2}},\ \bibinfo {pages} {033255} (\bibinfo {year} {2020})}\BibitemShut
  {NoStop}%
\end{thebibliography}
\end{document}